\documentclass[12pt]{article}
\usepackage[backend=bibtex, style=numeric]{biblatex}
\usepackage{booktabs, ragged2e}
\usepackage{showlabels}
\usepackage{amsmath, amsthm, amssymb,amscd, bbm}
\usepackage{fullpage, paralist,enumerate}
\usepackage{verbatim}
\usepackage[all]{xy}
\usepackage[parfill]{parskip}
\usepackage{graphicx}
\usepackage{subcaption}
\usepackage{mathabx}
\usepackage{float}
\usepackage{authblk}
\CompileMatrices

\usepackage[shortlabels] {enumitem}
\usepackage{xcolor}
\usepackage{hyperref}
\hypersetup{
    colorlinks   = true,
    citecolor    = blue,
    urlcolor=blue,
}

\newcommand{\R}{{\mathbb R}}

\newcommand{\N}{{\mathbb N}}

\newcommand{\F}{{\mathcal F}}

\newcommand{\bP}{{\mathbb P}}

\newcommand{\cF}{{\mathcal F}}

\newcommand{\abs}[1]{\left| #1 \right|}

\newtheorem{theorem}{Theorem}[section]
\newtheorem{lemma}[theorem]{Lemma}
\newtheorem{proposition}[theorem]{Proposition}

\theoremstyle{remark}

\theoremstyle{remark}

\newcommand{\norm}[1]{\left|\left|#1\right|\right|}

\newcommand{\argmin}{\mathrm{argmin }}

\newcommand{\TV}{\operatorname{TV}}
\newcommand{\inc}{\operatorname{inc}}

\DeclareMathOperator{\diag}{diag}
\DeclareMathOperator{\BV}{BV}
\DeclareMathOperator{\st}{s.t.}

\title{Fused Density Estimation: Theory and Methods}

\author[1]{Robert Bassett}
\author[2]{James Sharpnack}
\affil[1]{Department of Operations Research, Naval Postgraduate School}
\affil[2]{Department of Statistics, UC Davis}

\begin{document}
\maketitle

\begin{abstract}
In this paper we introduce a method for nonparametric density estimation on infrastructure networks. We define \emph{fused density estimators} as solutions to a total variation regularized maximum-likelihood density estimation problem. We provide theoretical support for fused density estimation by proving that the squared Hellinger rate of convergence for the estimator achieves the minimax bound over univariate densities of log-bounded variation. We reduce the original variational formulation in order to transform it into a tractable, finite-dimensional quadratic program.  Because random variables the networks we consider generalizations of the univariate case, this method also provides a useful tool for univariate density estimation. Lastly, we apply this method and assess its performance on examples in the univariate and infrastructure network settings. We compare the performance of different optimization techniques to solve the problem, and use these results to inform recommendations for the computation of fused density estimators.

\end{abstract}

\section{Introduction}

In the pantheon of statistical tools, the histogram remains the primary way to explore univariate empirical distributions.
Since its introduction by Karl Pearson in the late 19th century, the form of the histogram has remained largely unchanged.
In practice, the regular histogram, with its equal bin widths chosen by simple heuristic formulas, remains one of the most ubiquitous statistical methods.
Most methodological improvements on the regular histogram have come from the selection of bin widths---this includes varying bin widths to construct  irregular histograms---motivated by thinking of the histogram as a piecewise constant density estimate.
In this work, we study a piecewise constant density estimation technique based on total variation penalized maximum likelihood. We call this method fused density estimation (FDE).
We extend FDE from irregular histogram selection to density estimation over geometric networks, which can be used to model observations on infrastructure networks like road systems and water supply networks.
The use of fusion penalties for density estimation is inspired by recent advances in theory and algorithms for the fused lasso over graphs \cite{padilla,Sharpnack}.
Our thesis, that FDE is an important algorithmic primitive for statistical modeling, compression, and exploration of stochastic processes, is supported by our development of fast implementations, minimax statistical theory, and experimental results.

In 1926, \cite{sturges1926choice} provided a heuristic for regular histogram selection where, naturally, the bin width increases with the range and decreases with the number of points.
The regular histogram is an efficient density estimate when the underlying density is uniformly smooth, but irregular histograms can `zoom in' to regions where there is more data and better capture the local smoothness of the density.
A simple irregular histogram, known as the equal-area histogram, is constructed by partitioning the domain so that each bin has the same number of points.
\cite{denby2009variations} noted that the equal-area histogram can often split bins unnecessarily when the density is smooth and merge bins when the density is variable, and proposed a heuristic method to correct this oversight. 
Recently, \cite{li2016essential} proposed the essential histogram, an irregular histogram constructed such that it has the fewest number of bins and lies within a confidence band of the empirical distribution. While theoretically attractive, in practice its complex formulation is intractable and requires approximation. If the underlying density is nearly constant over a region, then the empirical distribution is well approximated locally by a constant, and hence the essential histogram will tend to not split this region into multiple bins.
Such a method is called {\em locally adaptive}, because it adapts to the local smoothness of the underlying density.

In Figure \ref{fig:histcomp}, we compare FDE to the regular histogram, both of which have 70 bins.
Because FDE can be thought of as a bin selection procedure, in this example, we recompute the restricted MLE after the bin selection, which is common practice for model selection with lasso-type methods.
We see that with 70 bins the regular histogram can capture the variability in the left-most region of the domain but under-smooths in the right-most region.
We can compare this to FDE which adapts to the local smoothness of the true density.
As a natural extension of 1-dimensional data, we will consider distributions that lie on geometric networks---graphs where the edges are continuous line segments---such as is common in many infrastructure networks.
Another motivation to use total variation penalties is that they are easily defined over any geometric network, in contrast to other methods, such as the essential histogram and multiscale methods.
Figure \ref{fig:SanDiego} depicts the FDE for data in downtown San Diego. The geometric network is generated from the road network in the area, and observations on the geometric network are the locations of eateries (data extracted from the OpenStreetMap database \cite{OSM}).

\begin{figure}
    \centering
    \includegraphics[width=.9\textwidth]{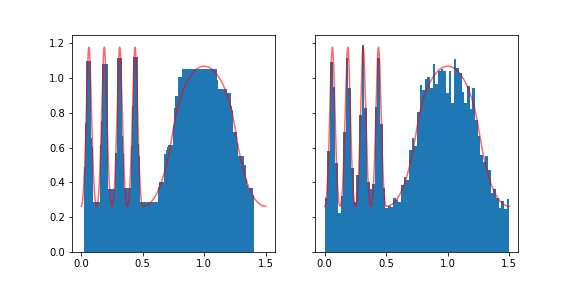}
    \caption{A comparison of FDE (left) and the regular histogram (right) of 10,000 data points from a density (red) with varying smoothness---both have 70 bins. }
    \label{fig:histcomp}
\end{figure}

\begin{figure}[t]
\centering
\includegraphics[width = .9\textwidth]{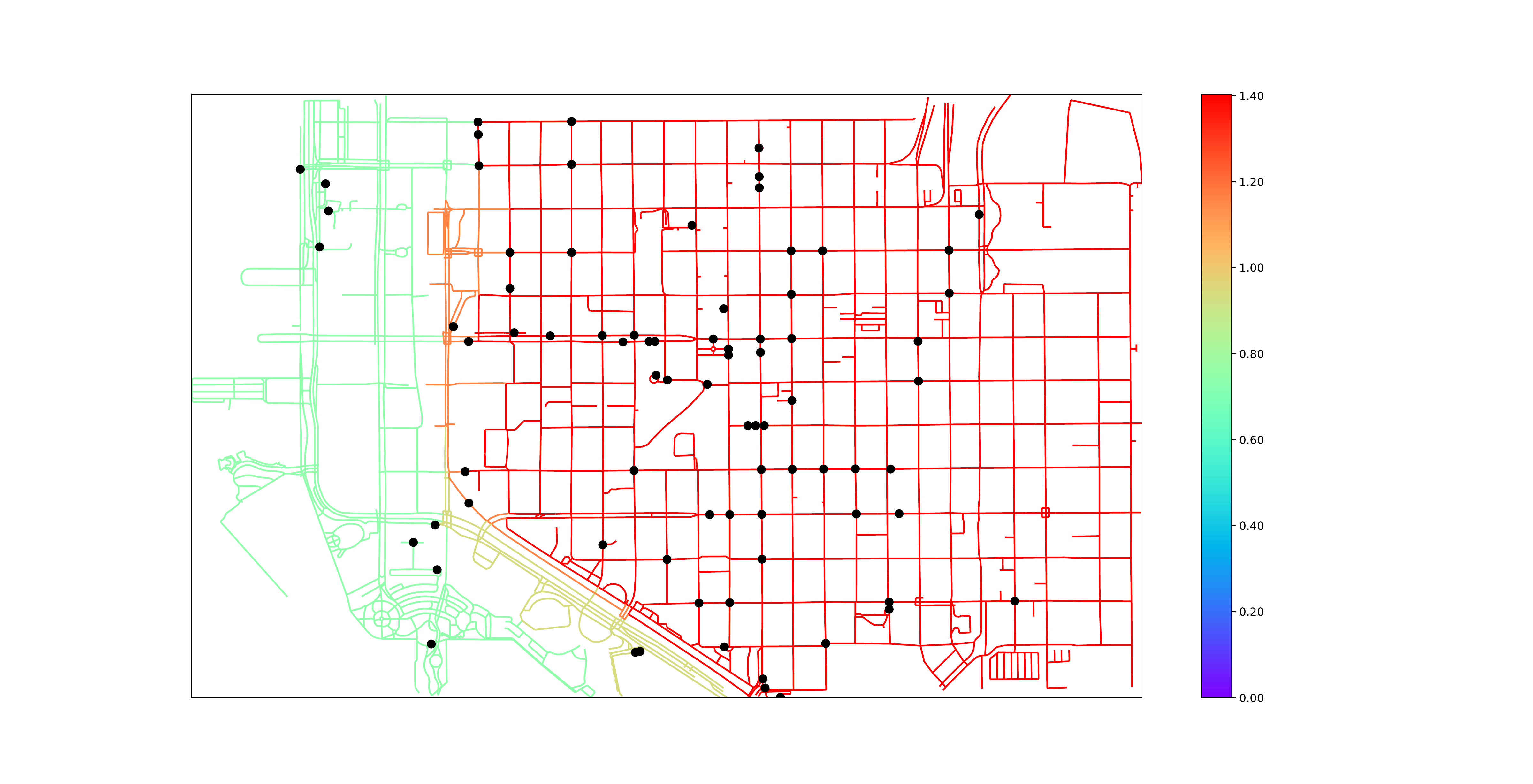}
\vspace{-.5cm}
\caption{FDE for the location of eateries in downtown San Diego.}
\label{fig:SanDiego}
\end{figure}

Without any constraints, maximum likelihood will select histograms that have high variation (as in Figure \ref{fig:histcomp}), so to regularize the problem, we bias the solution to have low total variation.
Total variation penalization is a popular method for denoising images, time series, and signals over the vertices of a graph, with many modern methods available for computation, such as alternating direction method of multipliers, projected Newton methods, and split Bregman iteration \cite{tibshirani2005sparsity,rudin1992nonlinear,beck2009fast,Sharpnack}. 
Distributions over geometric networks, which we consider here, are distinguished from this literature by the fact that observations can occur at any point along an edge of the network. 
This leads to a variational density estimation problem, which we reduce to a finite dimensional formulation.

\noindent {\bf Contribution 1.} We show that the variational FDE is equivalent to a total variation penalized weighted least squares problem enabling fast optimization.

In order to justify the use of FDE, we will analyze the statistical performance of FDE for densities of log-bounded variation over geometric networks.
The majority of statistical guarantees for density estimates control some notion of divergence between the estimate and the true underlying density.
Several authors have used the $L_2$ loss (mean integrated square error) to evaluate their methods for tuning the bin width for the regular histogram \cite{scott1979optimal, freedman1981histogram, birge1997model}. 
While it is appealing to use $L_2$ loss, it is not invariant to choice of base measure, and divergence measures such as $L_1$, Hellinger loss, and the Kullback-Leibler (KL) divergence are preferred for maximum likelihood---an idea pioneered by Le Cam \cite{le2012asymptotics} and furthered by \cite{devroye1985nonparametric, hall1988minimizing}.
By appealing to Hellinger loss, \cite{birge2006many} proposed a method for optimal choice of the number of bins in a regular histogram, and we will similarly focus on Hellinger loss.

\noindent {\bf Contribution 2.} We provide a minimax non-parametric Hellinger distance rate guarantee for FDE in the univariate case, over densities of log-bounded variation.


When the $\log$-density lies in a Sobolev space, an appropriate non-parametric approach to density estimation is maximum likelihood with a smoothing splines penalty \cite{silverman}.
The smoothing spline method is not locally adaptive because it does not adjust to the local smoothness of the density or $\log$-density.
Epi-splines, \cite{RoysetWets}, are density estimates formed by maximizing the likelihood such that the density, or log-density, has a representation in a local basis and lies in a prespecified constraint set.
\cite{donoho1996density} and \cite{koo1996wavelet} studied wavelet thresholding for density estimation and proved $L_p$ rate and KL-divergence guarantees respectively.
In a related work, \cite{koo2000logspline} considered log-spline density estimation from binned data with stepwise knot selection.
\cite{willett2007multiscale} used a recursive partitioning approach to form adaptive polynomial estimates of the density, a similar approach to wavelet decomposition.

Total variation penalization has previously been proposed as a histogram regularization technique in \cite{koenker2007density, sardy2010density, padilla2015nonparametric}.
Particularly, \cite{padilla2015nonparametric} separately studies the variational form of the fused density estimate and a discrete variant, and provides theoretical guarantees for Lipschitz classes.
Our computational results improve on these works by minimizing a variational objective directly,  instead of separately proposing discrete approximations to the variational problem. Our theoretical analysis improves on previous work by studying total variation classes directly by showing that FDE achieves the minimax rate for Hellinger-risk over all densities of log-bounded variation. Moreover, we consider density estimation on geometric networks, and extend our Hellinger rate guarantees to this novel setting.

%


\noindent {\bf Contribution 3.} We prove that the same Hellinger distance rate guarantee for the univariate case also holds for any connected geometric network.


\subsection{Problem Statement}

When considering road systems and water networks, we observe that individual roads or pipes can be modeled as line segments, and the entire network constructed by joining these segments at nodes of intersection.
Mathematically, we model this as a \emph{geometric network} $G$, a finite collection of nodes $V$ and edges $E$, where each edge is identified with a closed and bounded interval of the real line. 
Each edge in the network has a well-defined notion of length, inherited from the length of the closed interval.  
We fix an orientation of $G$ by assigning, for each edge $e = \{v_{i}, v_{j}\}$, a bijection between $\{v_{i},v_{j}\}$ and the endpoints of the closed interval associated with $e$.  
This corresponds to the intuitive notion of ``gluing'' edges together to form a geometric network. Because we only discuss geometric networks in this paper, we will often refer to them as networks.

A \emph{point} in a geometric network $G$ is an element of one of the closed intervals identified with edges in $G$, modulo the equivalence of endpoints corresponding to the same node. After assigning an orientation to the network, a point can be viewed as a pair $(e,t)$, where $e$ is an edge and $t$ is a real number in the interval identified with $e$. However, because we wish to emphasize the network as a geometric object in its own right, we will only use this notation when our use of univariate theory makes it necessary.
\begin{figure}[ht]
\hspace{-.75cm}
\begin{subfigure}{.55\textwidth}
  \centering
  \includegraphics[width=\linewidth]{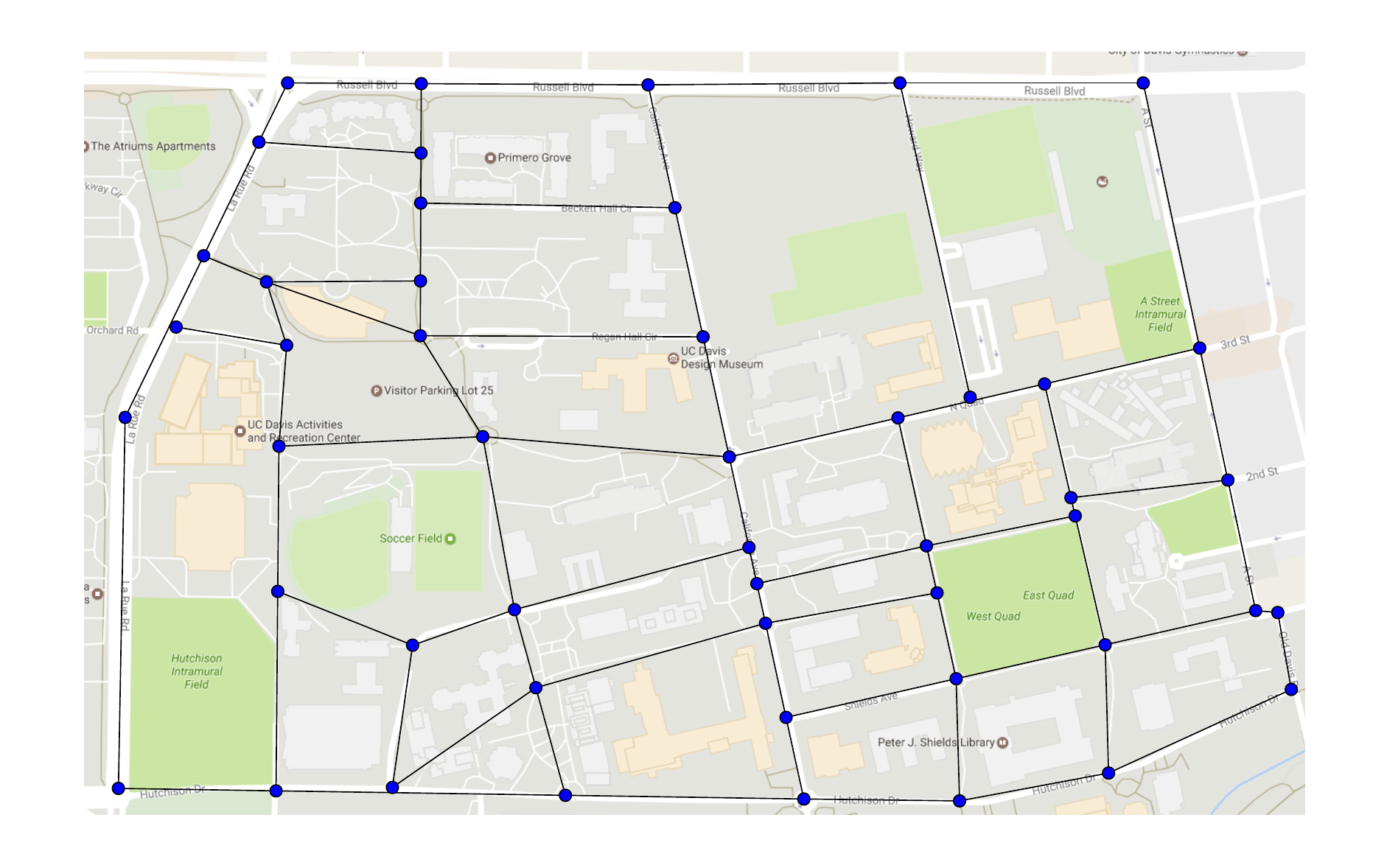}
\end{subfigure}%
\begin{subfigure}{.55\textwidth}
  \centering
  \includegraphics[width=\linewidth]{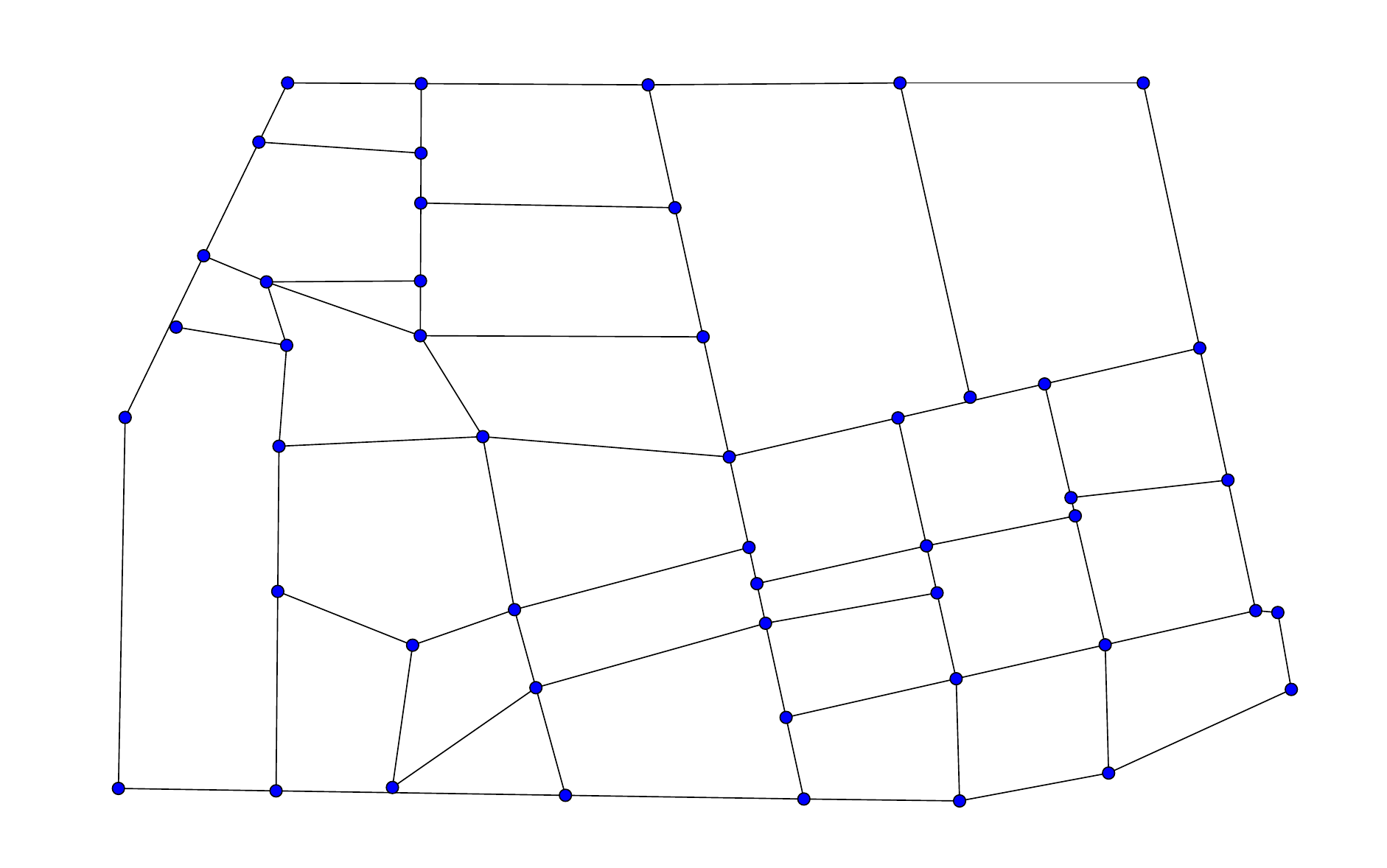}
\end{subfigure}
\caption{An example of a geometric network. This network is formed from bike paths on a university campus. Nodes identify intersections of paths. Edges are paths connecting these intersections, and the length of each edge is the length of the corresponding path. In forming the geometric network, we discard all information related to its embedding in $\R^2$, only preserving the network structure and path lengths.}
\end{figure}

A real-valued function $g$, defined on a geometric network $G$, is a collection of univariate functions $\{g_{e}\}_{e \in E}$, defined on the edges of $G$. We require that the function respects the network structure, by which we mean that for any two edges $e_1$ and $e_2$ which are incident at a node $v$, $g_{e_{1}}(v) = g_{e_{2}}(v)$. We abuse notation slightly--by referring to $g_{e}(v)$, we mean $g$ evaluated at the endpoint of the interval identified to $v$. A geometric network $G$ inherits a measure from its univariate segments in a natural way, as the sum of the Lebesgue measure along each segment. With this measure we have a straight-forward extension of the Lebesgue measure to $G$, making $G$ a measurable space. 

For any random variable taking values on the network $G$, we will assume that the measure induced by the random variable is absolutely continuous with respect to the base measure, $dx$, and so has density $f$.  We will abuse notation by using $dx$ to refer to both the Lebesgue measure and the base measure on a geometric graph; which of these we mean will be clear from its context. Furthermore, we assume that the density is non-zero everywhere, so that its logarithm is well defined. Moreover, we will assume the log-density is not arbitrarily variable, and for this purpose we will use the notion of total variation.
Let $B \subseteq \mathbb{R}$. The \emph{total variation} of a function $g: B \to \R$ is defined as
$$\TV(g) = \sup_{P \subset B} \sum_{z_{i} \in P} \abs{g(z_{i}) - g(z_{i+1})}.$$
The supremum is over all partitions, or finite ordered point-subsets $P$, of $B$. 
For a real-valued function $g$ defined on a network $G$, we extend the univariate definition to
$$\TV(g) = \sum_{e \in E} \TV(g_{e}).$$
One advantage of the use of the $\TV$ penalty is that is it invariant to the choice of the segment length in the geometric network, so scaling the edge by a constant multiplier leaves the total variation unchanged. As a consequence fused density estimation will be invariant to the choice of edge length.

Let $f_{0}$ be a density on a geometric network $G$, and $x_{1},...,x_{n}$ an independent sample identically distributed according to $f_{0}$. 
Let $P_{n} = \frac{1}{n} \sum \delta_{x_{i}}$ be the empirical measure associated to the sample. 
We let $P_n$ act on a function, by which we mean that we take the expectation of that function with respect to $P_n$. So for any function $f$,
$$P_{n}(f) = \int f \, d P_{n} = \frac{1}{n} \sum_{i=1}^{n} f(x_{i}).$$
We will also use $P(f)$ to denote $\int f \, dP$ for non-empirical measures $P$.

Fix $\lambda \in \mathbb{R}^{+}$. A \emph{fused density estimator} (FDE) of $f_{0}$ is a density $\hat{f} = \exp(\hat g)$, such that the log-density $\hat{g}$ is minimizer or the following program,
\begin{equation} \label{fde}
\min \; -P_{n}(g) + \lambda \TV(g) \text{ s.t. } \int e^{g} \, dx = 1
\end{equation}
where the minimum is taken over all functions $g: G \to \mathbb{R}$ for which the expression is finite and the resulting $f$ is a valid density. 
That is, $f \in \mathcal{F}$ and $g \in \mathcal{G}$ where
$$\mathcal{F} = \{e^{g}: g \in \mathcal{G}\}, \; \; \; \mathcal{G} = \left\{g : \TV(g) < \infty, \; \; \int_{G} e^{g} \, dx = 1 \right\}.$$
The set $\mathcal{F}$ will be referred as the set of densities with \emph{log-bounded variation}. Indeed, the integration constraint on elements of $\mathcal{G}$ makes them log-densities.
Note that densities in $\mathcal{F}$ are necessarily bounded above and away from zero, as a result of the total variation condition.

The program in \eqref{fde} is variational, because it is a minimizer over an infinite dimensional function space.
It is quite common for variational problems in non-parametric statistics to involve a reproducing kernel Hilbert space (RKHS) penalty, as opposed to a total variation penalty \cite{wahba1990spline}.
In the RKHS setting setting, the Hilbert space allows us to establish representer theorems, which reduce the variational program to an equivalent finite dimensional one, so that it can be solved numerically.  
The space of functions of bounded variation, on the other hand, is an example of a more general Banach space, so RKHS results cannot be applied to this setting. In the next section, we discuss representer theorems for \eqref{fde}, and further show that it can be solved using a sparse quadratic program.

\section{Computation} \label{Computation}

In this section we provide results toward the computation of fused density estimators. The key challenge is the variational formulation of the Fused Density Estimator \eqref{fde}. 
To this end, we prove that solutions to the variational problem can be finitely parametrized. Moreover, we show that after applying this representer theorem, the finite-dimensional analog of \eqref{fde} has an equivalent formulation as a total variation penalized least-squares problem. 
Our main theorem of this section, which reduces the computation of a fused density estimator to a weighted fused-lasso problem, follows.

\begin{theorem}[Informal] \label{MainCompThm}
For $\lambda>\frac{1}{2n}$ the FDE exists almost surely. It can be computed as the minimizer to a finite-dimensional convex, sparse, and total variation penalized quadratic program. That is, FDEs are solutions to an optimization of the form
\begin{equation}\label{InformalProb}
\min_{z\in \R^d} \frac{1}{2} z^{T} P z + a^{T} z + \norm{Dz}_{1}.
\end{equation}
The details of this theorem, by which we mean the constructions of $P$, $D$, $a$, and the connection between the minimizer $\hat{z}$ of \eqref{InformalProb} and the FDE $\hat{f}$, are given later in this section.
\end{theorem}

Theorem \ref{MainCompThm} demonstrates that the FDE \eqref{fde} can be computed as a specific incarnation of the generalized lasso, for which there are well known fast implementations \cite{arnold2016efficient}.
In practice we will solve the dual to this problem, which we discuss in Theorem \ref{Dual}.
Theorem \ref{RestateThm} is a precise restatement of Theorem \ref{MainCompThm}. In order to prove it, we proceed through a series of important lemmas. Lemma \ref{UnCon} proves that minimizers of \eqref{fde} exist almost surely for $\lambda > 1/2n$, below which the solution degenerates to dirac masses at observations. The almost surely qualification pertains to the maximum number of observations which occur at any single point of the network; if observations occur simultaneously then $\lambda$ must be increased to overcome degeneracy. Lemma \ref{UnCon} further transforms the FDE problem from constrained to unconstrained by removing the integration constraint. From this new formulation, Lemma \ref{FinDim} shows that the search space for the fused density estimator problem can be reduced from functions of bounded variation to an equivalent, finite-dimensional version. Theorem \ref{RestateThm} performs the final step in the proof--demonstrating that the previously derived finite-dimensional problem can be solved using a $\ell_1$ penalized quadratic program. The last subsection in this section is tangential, but sheds further light on the structure of fused density estimators. Proposition \ref{OrderingProp}, which we refer to as the Ordering Property, qualifies the local-adaptivity of fused density estimators by describing their local structure. Omitted proofs can be found in Appendix \ref{AppA}.

\subsection{Main Computational Results}

Our first lemma reduces the fused density estimator problem, \eqref{fde}, to an unconstrained program where the integral constraint is incorporated into the objective. 
This result is originally due to Silverman \cite{silverman}, who proved the result in the context of univariate density estimation and Sobolev-norm penalties. 
Minor modifications allow us to extend it to geometric networks and the non-Sobolev total variation penalty.
\begin{lemma}\label{UnCon}
The problem
\begin{equation} \label{UnConProb}
\min_{g \in \mathcal{G}} -P_{n}(g) + \lambda \TV(g) + \int_{G} e^{g} \, dx
\end{equation}
gives an equivalent formulation of \eqref{fde}, because minimizers $\hat{g}$ of \eqref{UnConProb} satisfy $\int_{G} e^{\hat{g}} \, dx = 1$.
\end{lemma}

We remark that the objective in Lemma \ref{UnCon} is equivalent to total variation penalized Poisson process likelihood, where the log-intensity is $g$, so our computations also apply to that setting.
Lemma \ref{UnCon} gives that the fused density estimator definition \eqref{fde} can instead be solved by the unconstrained problem \eqref{UnConProb} over all functions $g$ on $G$ of bounded variation. An alternative interpretation of the lemma is that the Lagrange multiplier associated to the constraint in \eqref{fde} is $1$.
The next lemma reduces the unconstrained problem \eqref{UnConProb} to an equivalent finite-dimensional version. The proof technique is analogous to similar results in \cite{MammenvandeGeer}. In the context of Reproducing Kernel Hilbert Spaces, results that reduce variational problem formulations to finite-dimensional analogs are referred to as \emph{representer theorems}, eg. \cite{wahba1990spline}. We will also use this language to describe our result, even though we are in a more general Banach space setting. The result demonstrates that FDEs have large, piecewise constant regions, which is a well known property of fusion penalties \cite{tibshirani2005sparsity, Kim, Sharpnack}.

\begin{lemma}[Representer Theorem] \label{FinDim}
A fused density estimator $\hat{f}$ must be piecewise constant along each edge. All discontinuities are contained in the set $\{x_{1},...,x_{n}\} \cup V$, the observations and the nodes of $G$.
\end{lemma}

Using Lemma \ref{FinDim}, we can parametrize fused density estimators with three finite-dimensional vectors: the fused density estimator at the observation points, $p$, the fused density estimator at the vertices of $G$, $k$, and the piecewise constant values of the fused density estimator, $c$. 
For simplicity, we will assume that no two observations occur at the same location, a condition that we can and will relax in the remark following Theorem \ref{RestateThm}.

Let $n_{e}$ denote the number of observations along edge $e$. We will denote by $p_{e,i}$ the value, in the vector $p$, of the $i$th ordered observation along edge $e$. For a FDE $f = e^g$, $c_{e,i}$ is the value taken by $g$ between the $(i-1)th$ and $i$th observation in the interval associated with $e$, where the $0$th and $(n_{e}+1)$th observations are set to be the endpoints of that interval. Similarly, $s_{e,i}= x_{e,i} - x_{e,i-1}$ with the convention that $x_{e,0}$ and $x_{e,n_{e}+1}$ are the endpoints of the interval. This gives the length of the segment between two observations, over which the FDE is piecewise constant.
We denote by $k_v$ the value in $k$ at the vertex $v$. For a given node $v$, let $\inc(v)$ denote the set of edges which are incident to $v$ and denote by $c_{e,v}$ the segment in $c$ which is incident to $v$. The problem \eqref{UnConProb} becomes
\begin{align*}
\min_{p,c,k} & \sum_{e \in E} \left\{ -\frac{1}{n} \sum_{i=1}^{n_{e}} p_{e,i} + \lambda \sum_{i=1}^{n_{e}} \abs{p_{e,i} - c_{e,i}} + \abs{p_{e,i}-c_{e,i+1}} +\sum_{i=1}^{n_{e}+1} s_{e,i} e^{c_{e,i}} \right\} + \lambda \sum_{v \in V} \sum_{e \in \inc(v)} \abs{k_{v}-c_{e,v}}
\end{align*}
The first summand, over the edges in $E$, gives the log-likelihood term, the total variation along an edge, and the integration term. The second summand gives the total variation at nodes of the geometric graph. 

Let $F$ denote the objective function above. We show that this problem can be further reduced by removing the $p_{e,i}$ variables. Indeed, for any vectors $\hat{c}, \hat{k}$, let $\tilde{F}(p) = F(p, \hat{c}, \hat{k})$. A necessary condition for any $\hat{p}, \hat{c}, \hat{k}$ to minimize $F$ is $\hat{p} \in \argmin_{p} \; \tilde{F}(p)$. But $\tilde{F}(p)$ does not have a lower bound when $\lambda < \frac{1}{2n}$. Furthermore, the set $\argmin_{p} \; \tilde{F}(p)$ is unbounded when $\lambda = \frac{1}{2n}$, and $\tilde{F}(p)$ has a unique minimum when $\lambda > \frac{1}{2n}$. These facts are clear from the graph of the functions $p_{e,i} \to -p_{e,i}/n + \lambda(\abs{p_{e,i} - c_{e,i}} + \abs{p_{e,i} - c_{e,i+1}}$, which occur as summands in $F$. The function is given in Figure \ref{fig:p_picture}.
\begin{figure}
\centering 
\includegraphics[width=12cm]{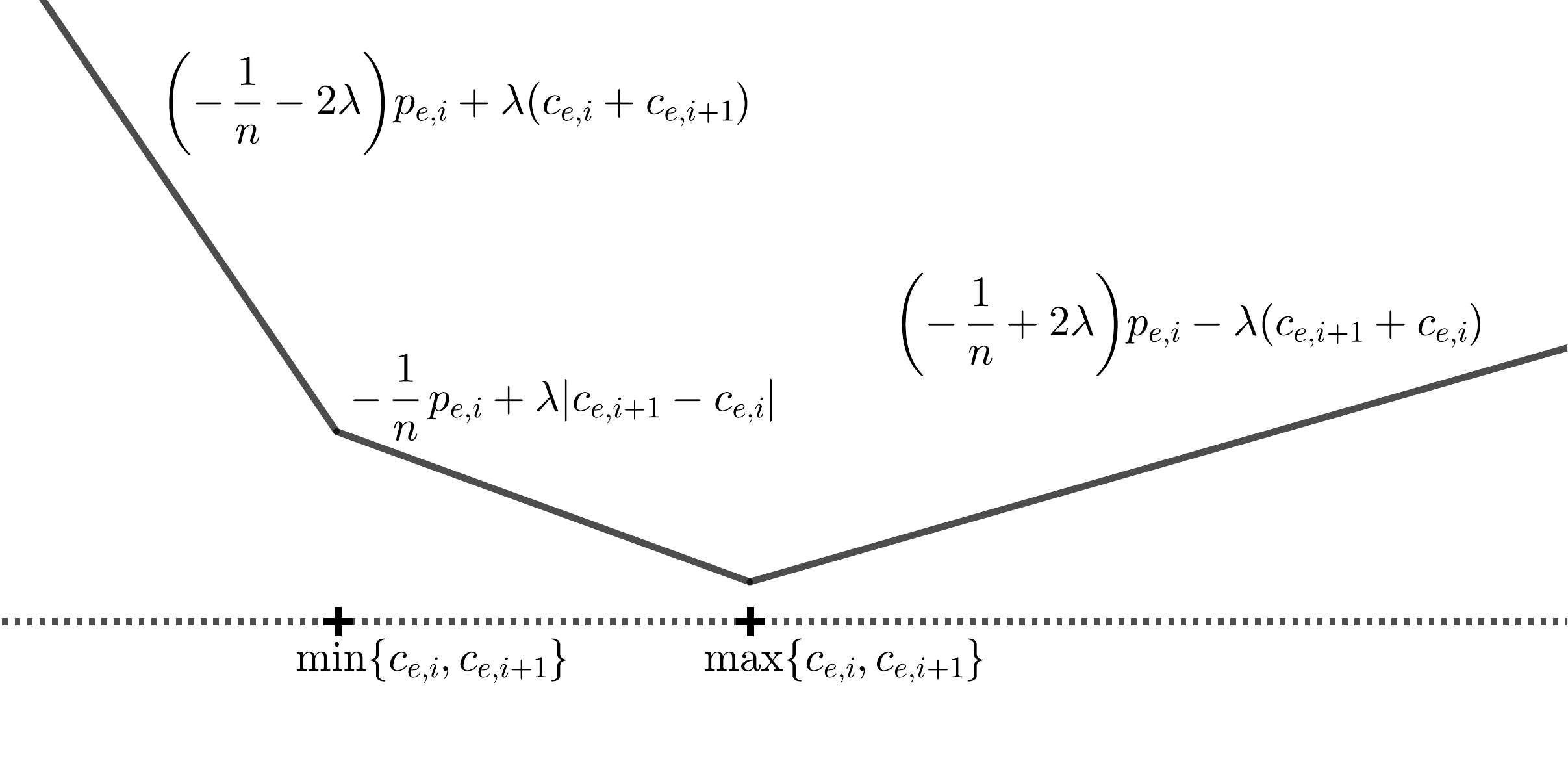}
\caption{The function $p_{e,i} \to -p_{e,i}/n + \lambda(\abs{p_{e,i} - c_{e,i}} + \abs{p_{e,i} - c_{e,i+1}})$ as the maximum of three affine functions. It attains a unique minimum at $p_{e,i} = \max\{c_{e,i}, c_{e,i+1}\}$ when $\lambda > \frac{1}{2n}$.}
\label{fig:p_picture}
\end{figure}
It is the maximum of three affine functions, from which we conclude that the minimum of $\tilde{F}$ is attained uniquely at $p_{e,i} = \max\{c_{e,i}, c_{e,i+1}\}$ when $\lambda > \frac{1}{2n}$. We have shown that $\lambda = \frac{1}{2n}$ is a critical point for the existence of FDEs, below which the total variation penalty is not strong enough to prevent degenerate solutions to \eqref{fde}. For $\lambda > \frac{1}{2n}$, the value of an FDE at observations is well-behaved and we can reduce $F(p, c, k)$ to

\begin{align*}
\min_{c,k} & \sum_{e \in E} \left\{ -\frac{1}{n} \sum_{i=1}^{n_{e}} \max\{c_{e,i}, c_{e,i+1}\} + \lambda \sum_{i=1}^{n_{e}} \abs{c_{e,i}-c_{e,i+1}} +\sum_{i=1}^{n_{e}+1} s_{e,i} e^{c_{e,i}} \right\} + \lambda \sum_{v \in V} \sum_{e \in \inc(v)} \abs{k_{v}-c_{e,v}}.
\end{align*}
Because 
$2 \cdot \max\{c_{e,i}, c_{e,i+1}\} = c_{e,i}+c_{e,i+1}+\abs{c_{e,i}-c_{e,i+1}},$
we have the further equivalence,
\begin{align*}
\min_{c,k} \sum_{e \in E} \left\{ -\frac{1}{2n} \sum_{i=1}^{n_{e}}(c_{e,i}+c_{e,i+1}) + (\lambda- \frac{1}{2n})\sum_{i=1}^{n_{e}} \abs{c_{e,i}-c_{e,i+1}} + \sum_{i=1}^{n_{e}+1} s_{e,i} e^{c_{e,i}}\right\} + \lambda \sum_{v \in V} \sum_{e \in \inc(v)} \abs{k_{v}-c_{e,v}}.
\end{align*}

By \cite[Theorem 23.8]{rockafellar}, a necessary and sufficient condition for $\hat{c},\hat{k}$ to solve this problem is
\begin{align*} 
0 \in & \sum_{e \in E} \left\{ -\frac{1}{2n}\sum_{i=1}^{n_{e}} \partial \left(c_{e,i}+c_{e,i+1}\right) + (\lambda-\frac{1}{2n})\sum_{i=1}^{n_{e}} \partial\left(\abs{c_{e,i}-c_{e,i+1}}\right) + \sum_{i=1}^{n_{e}+1} \partial \left( s_{e,i} e^{c_{e,i}} \right) \right\} \\
&+ \lambda \sum_{v \in V} \sum_{e \in \inc(v)} \partial \left(\abs{k_{v}-c_{e,v}}\right).
\end{align*}

Here we make an important point. The subdifferential of each $(c_{e,i}+c_{e,i+1})$ term is constant, and the subdifferential of $\abs{c_{e,i}-c_{e,i+1}}$ is piecewise constant, \emph{depending only on the ordering} of the terms $c_{e,i}$ and $c_{e,i+1}$. Similarly, the subdifferential of the $\abs{k_{v}-c_{e,v}}$ term is piecewise constant and again only depends on the ordering of its terms. Lastly, the subdifferential of $s_{e,i} e^{c_{e,i}}$ is given by its gradient: the $(e,i)$th coordinate of the subdifferential is $s_{e,i} e^{c_{e,i}}$. 

Consider the transformation $z = e^{c}$, $h = e^{k}$. This transformation preserves ordering of elements of $\hat{c}$ and $\hat{k}$, so the subdifferential of each absolute value term is invariant under this transformation. Pursuing this line of reasoning gives the following theorem. In order to facilitate its statement, we briefly establish some notation. 

The total variation of a FDE $f$ on $G$, which has been parametrized into vectors $z$ and $h$, can be expressed as a sum of pairwise distances between values in $z$ and $h$. That is, there are sets $J_1$ and $J_2$ of index pairs such that
$$\TV(f) = \sum_{(i,j) \in J_1} \abs{z_{i} - z_{j}} + \sum_{(i,j) \in J_2} \abs{z_{i} - h_{j}}.$$
This formulation depends on the underlying graph structure and the locations of the observations. The right-hand side of this expression can be written as the $\ell_1$ norm of a vector $C_1 z + C_2 h$, where $C_1$ and $C_2$ are matrices with elements in $\{-1,0,1\}$, each having $\abs{J_1}+\abs{J_2}$ rows. We will use the matrices $C_1$ and $C_2$, which satisfy $\TV(f) = \norm{C_1 z + C_2 h}_{1}$ and $C_2$ is zero in its first $\abs{J_1}$ rows. Let $n_{i} = \abs{J_{i}}$ for $i \in \{1,2\}$. Let 
$$B = \left(\begin{array}{cc} (\lambda-1/2n) I_{n_1\times n_1} & 0_{n_1 \times n_2} \\ 0_{n_2 \times n_1} & \lambda I_{n_{2} \times n_{2}} \end{array} \right).$$
and let $D_1$ and $D_2$ be the matrices $BC_1$ and $BC_2$, respectively. We denote by $x_{1},...,x_{n}$ the locations of observation on $G$, and we further partition these into ordered observations along each edge, so that $x_{e,i}$ denotes the $i$th observation along edge $e$. Recall the definition of $s_{e,i} = x_{e,i}-x_{e,i+1}$, and let $S$ be a diagonal matrix with $s$ on its diagonal. Lastly, define the vector $w$ such that
$$w_{e,i} = \begin{cases} -\frac{1}{2n} & i = 1  \text{ or } i=n_{e}+1\\ -\frac{1}{n} & \text{ otherwise.}\end{cases}.$$

\begin{theorem} \label{RestateThm}
Let $\lambda > \frac{1}{2n}$. Then the fused density estimator exists almost surely. It can be computed as follows. Let $z$ be a vector with indices enumerating the constant portions of the fused density estimator $\hat{f}$, such that $z_{e,i}$ denotes the value of the fused density estimator on the open interval between $x_{e,i}$ and $x_{e,i-1}$, or between an observation and the end of the edge if $i=1$ or $n_{e}+1$.  Let $h$ be a vector with indices enumerating the nodes in $G$, such that $h_{v}$ denotes the value of $\hat{f}$ at node $v$. Then the fused density estimator $\hat{f}$ for this sample is the minimizer of
\begin{equation} \label{PrimalProb}
\min_{z,h} \; \frac{1}{2} z^\top S z + w^\top z + \norm{ D_1 z  + D_2 h}_{1}.
\end{equation}

\end{theorem}

\begin{proof}
The proof follows directly from the line of reasoning before the theorem's statement. Details can be found in Appendix \ref{AppA}.
\end{proof}

\emph{Remark.} The condition on $\lambda$ is an important one. As discussed previously, when $\lambda < \frac{1}{2n}$ the total variation penalty is not strong enough to balance the likelihood term and minimizers of \eqref{fde} are degenerate. The almost surely condition is simply a requirement that no two observations occur at the same location, and no observations occur at nodes of the geometric network. With a slight modification of the assumption on $\lambda$, Theorem \ref{RestateThm} can be extended to the setting where multiple observations are allowed at a single location. This extension also allows observations to occur at nodes of the geometric network. In practice, this extension may be useful when dealing with imperfect data, though we will not focus on it here because it is a measure zero event in the density estimation paradigm. For completeness, we include the extension as Theorem \ref{ExtendThm} of Appendix \ref{AppA}.

Methods for computing solutions to the problem in Theorem \ref{RestateThm}--a total-variation regularized quadratic program--are well established. As in \cite{Kim}, we rely on solving the dual quadratic program. For convenience, we write the dual problem as a minimization instead of its typical maximum formulation

\begin{proposition}\label{Dual}
The dual problem to \eqref{PrimalProb} is 
\begin{align}\label{DualProb}
\min_{y} & \quad \frac{1}{2} y^{\top} D_{1} S^{-1} D_{1}^{\top} y + w^{\top} S^{-1} D_{1}^{\top} y \nonumber\\
& \quad \norm{y}_{\infty} \leq 1 \\
& \quad D_{2}^{\top} y = 0. \nonumber
\end{align} 
The primal solution $\hat{z}$ can be recovered from the dual $\hat{y}$ through the expression
$$\hat{z} = -S^{-1}(D_{1}^\top \hat{y} +w).$$
\end{proposition}
A more general statement of Proposition \ref{Dual} which suits the more general statement of Theorem \ref{RestateThm}, can be found in Appendix \ref{AppA}.
It is worth noting that strong duality between the primal and dual problems in \eqref{PrimalProb} and \eqref{DualProb} follows immediately. Indeed, both are extended linear-quadratic programs in the sense of \cite{VaAn}. By Theorem 11.42 in \cite{VaAn}, strong duality holds, and in addition both the primal and dual problem attain their minimum values, respectively, if and only if \eqref{PrimalProb} is bounded. This is guaranteed by the assumption on $\lambda$ in Theorem \ref{RestateThm}. Furthermore, the fact that the minimum of \eqref{PrimalProb} is attained gives the existence of FDEs as asserted in Theorem \ref{RestateThm}.

\subsection{Additional Properties of FDEs}
In this section, we state a result on the local structure of an FDE and provide additional comments on its implementation details. The result is intuitive: along an edge, the value of piecewise constant segments is inversely related to the length of the segment, relative to adjacent segments. Since smaller segments suggest higher probability in the corresponding region, this property demonstrates local structure of the estimator which aligns with essential global behavior.

\begin{proposition}[Ordering Property] \label{OrderingProp}
Let $s_{e,i}$ and $s_{e,i+1}$ be the lengths of two segments interior to an edge $e$, in the sense that $2 \leq i \leq n_{e}-1$. Assume further that no two observations occur at the same location. Then $s_{e,i} \leq s_{e,i+1}$ implies that $\hat{z}_{e,i} \geq \hat{z}_{e,i+1}$. Similarly, $s_{e,i} \geq s_{e,i+1}$ implies $\hat{z}_{e,i} \leq \hat{z}_{e,i+1}$.
\end{proposition}


Up to this point in our analysis, we have discussed the computation of the FDE without consideration for preprocessing the data or postprocessing our resulting FDE. Since the computation and rates of convergence of the FDE represent the bulk of our contribution, we will maintain this perspective in the remainder of the paper. It is worth mentioning, however, that FDE is amenable to pre and postprocessing. Handling multiple observations at a single location in Theorem \ref{RestateThm} makes initial binning or minor discretizations of data (such as projecting observations \emph{onto} a geometric network) straightforward. Moreover, the FDE can be viewed exclusively as a method for generating adaptive bin widths, where the resulting bins can then be fit to the data as in a regular histogram. This approach performs model selection (via FDE) and model fit (via a post-selection MLE) of the histogram separately, and is common practice in model selection using lasso and related methods \cite{PostSelectLasso, GeneralPostSelect}. When FDE is used exclusively to find bins, it becomes a change point localization method, instead of a nonparametric density estimator as in its original formulation. Though FDE is amenable to these examples of pre and postprocessing, we will examine the FDE as a density estimator in the remaining sections.

We also make some suggestions into the selection of $\lambda$. The choice of $\lambda$ leads to a fixed number of piecewise constant portions of the fused density estimator. In this sense, the choice of the $\lambda$ parameter is analogous to choosing the number of bins in histogram estimation. One can tune this selection with information-criteria (IC) such as AIC or BIC by selecting the FDE over a grid of $\lambda$ values that minimizes the IC.
Each of these ICs requires the specification of the degrees of freedom, which can be set to the number of selected piecewise constant regions in the graph, as is done in the Gaussian case \cite{Sharpnack, tibshirani2012degrees}. Alternatively, one could use cross-validation as a selection criterion. Implementing cross-validation is often practical for large problems because the sparse QP in \eqref{DualProb} can be solved very quickly, as we will see in the next section.

\section{Experiments}
We have established a tractable formulation of the fused density estimator in \eqref{DualProb}. Quadratic programming is a mature technology, so computing FDEs via quadratic programming dramatically improves its computation. Quadratic program solvers designed to leverage sparsity in the $D_{1}$ and $D_{2}$ matrices allow the optimization portion of fused density estimation to scale to large networks and many observations.
 
In this section we compute FDEs on a number of synthetic and real-world examples.\footnote{These examples can be found at \href{https://github.com/rbassett3/FDE-Tools}{github.com/rbassett3/FDE-Tools}, which also includes a Python package for fused density estimation on geometric networks.} We evaluate the performance of different optimization methods and provide recommendations for solvers which implement those methods. To facilitate accessibility and customization of these tools, each of the solvers we consider is open source and compare favorably with commercial alternatives.

\subsection{Univariate Examples}

We first evaluate fused density estimators in the context of univariate density estimation--where the geometric network $G$ is simply a single edge connecting two nodes. The operator $D_1 + D_2$ is especially simple in this setting, corresponding to an oriented edge-incidence matrix of a chain graph. 
Figure \ref{fig:UniImage} contains fused density estimators of the standard normal, exponential, and uniform densities, each derived from $100$ sample points. The $\lambda$ parameter in these experiments was selected by $20$ fold cross-validation.

\begin{figure}[ht!]
    \centering
    \begin{subfigure}[b]{0.32\textwidth}
        \includegraphics[width=\textwidth]{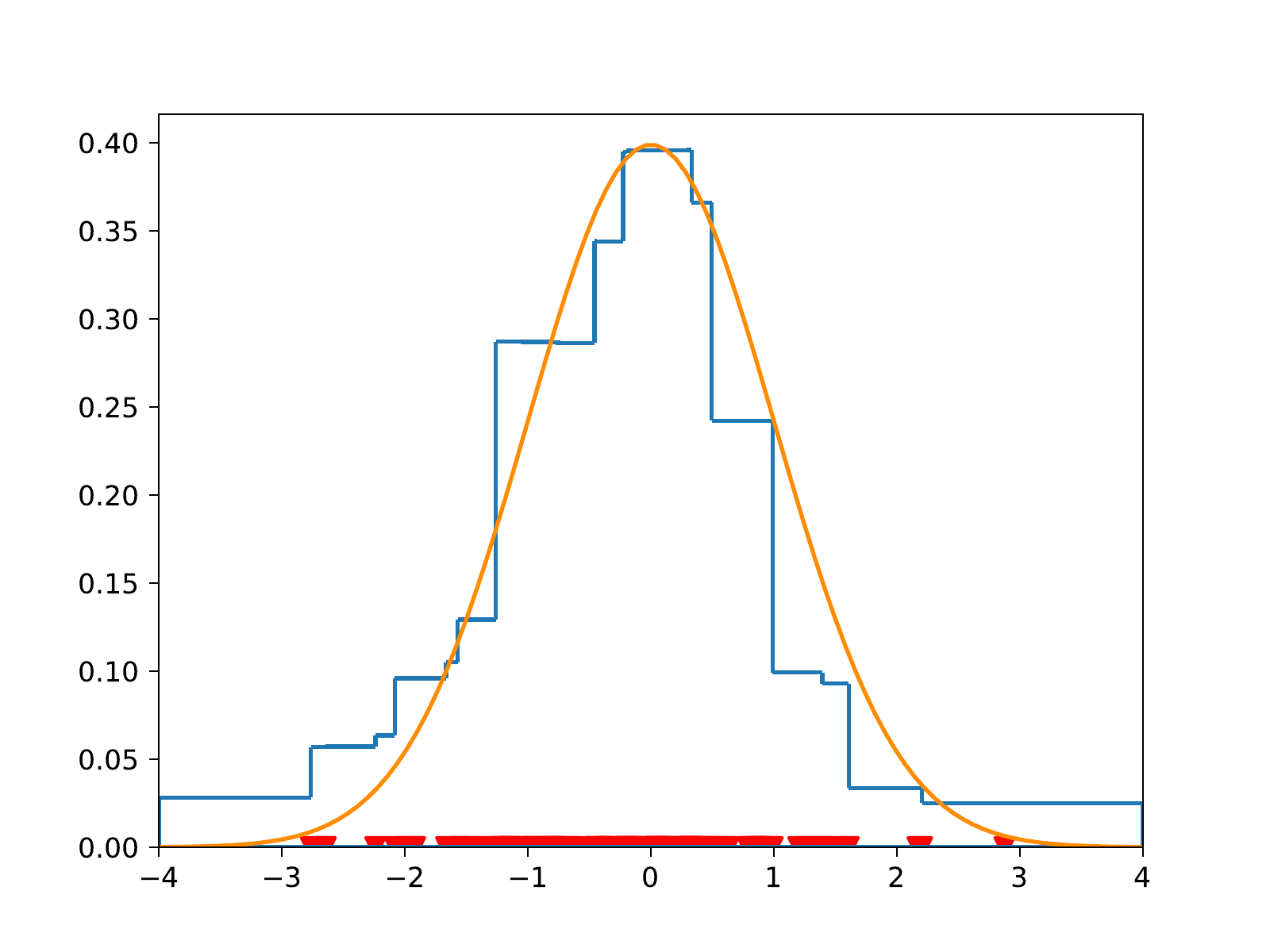}
    \end{subfigure}
    \begin{subfigure}[b]{0.32\textwidth}
        \includegraphics[width=\textwidth]{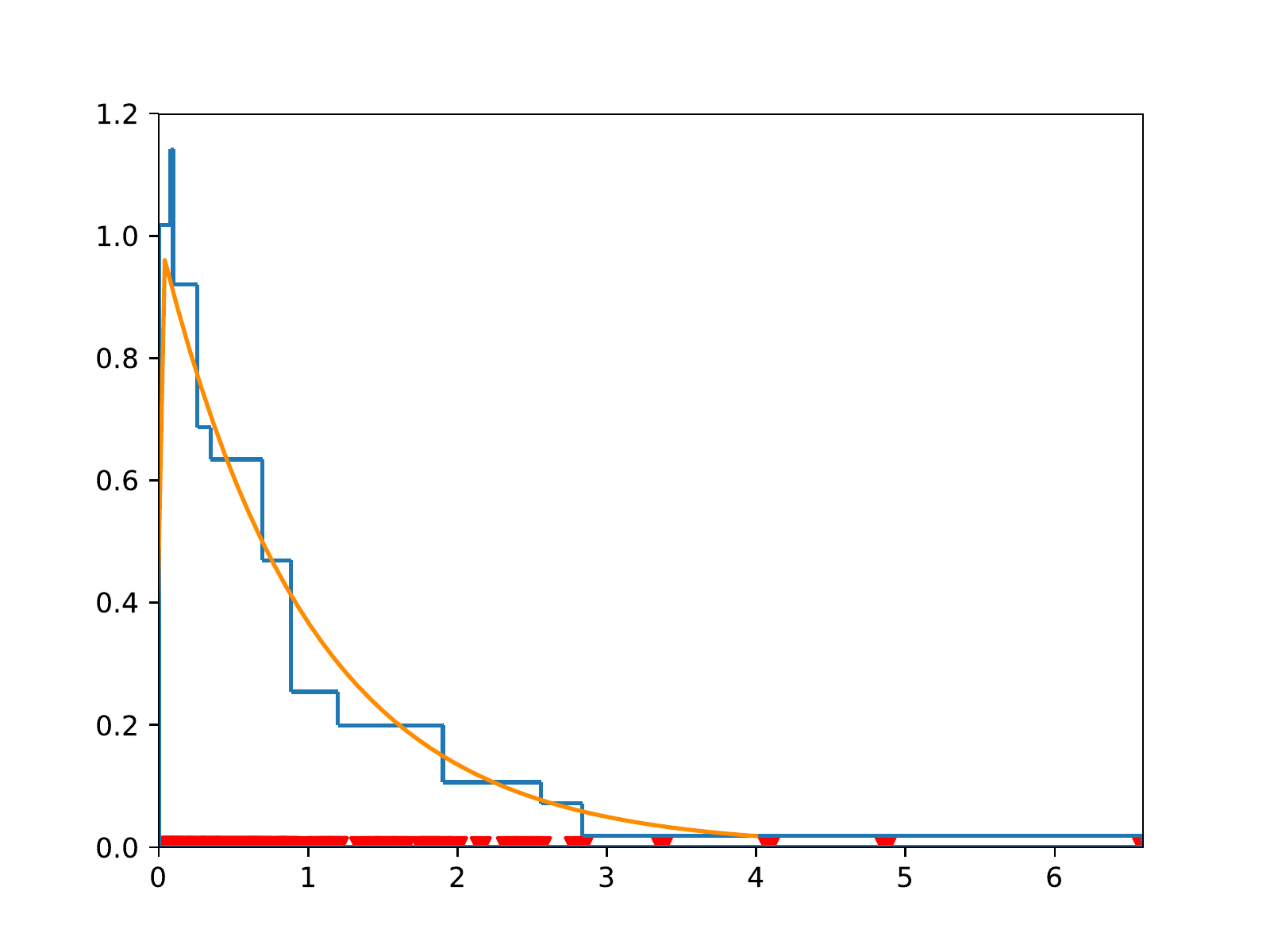}
    \end{subfigure}
    \begin{subfigure}[b]{0.32\textwidth}
        \includegraphics[width=\textwidth]{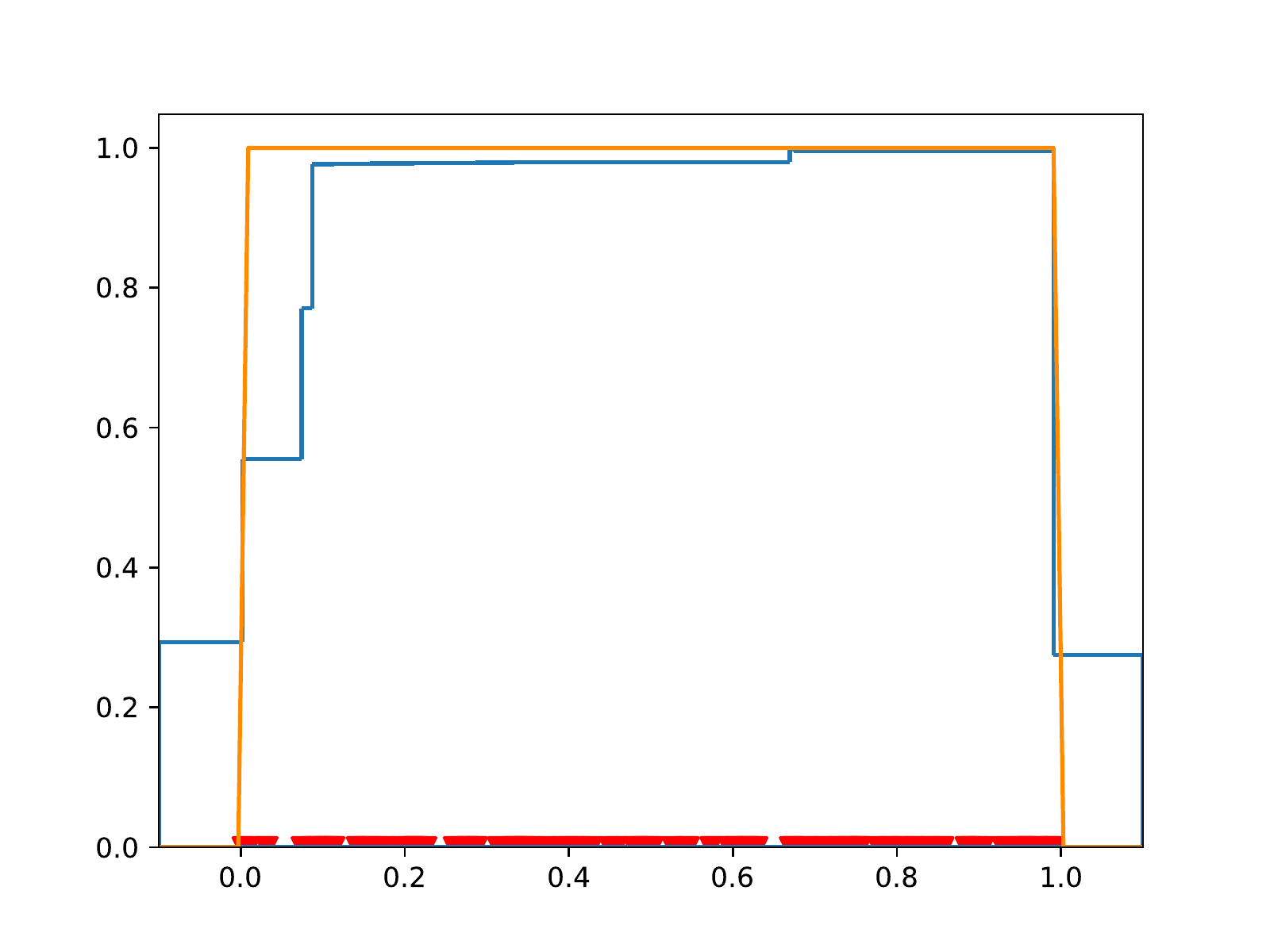}
    \end{subfigure}
    \caption{Univariate densities and fused density estimators}\label{fig:UniImage}
\end{figure}

\subsection{Geometric Network Examples}

We next evaluate FDEs on geometric networks. For each of these examples, the underlying geometric network is extracted from OpenStreetMap (OSM) database \cite{OSM}.
Figure \ref{fig:Baghdad} is a fused density estimator with domain taken to be the road network in a region of the city of Baghdad. Observations are the locations of terrorist incidents which occurred in this region from 2013 to 2016, according to the Global Terrorism Database \cite{GTD}. The density we attempt to infer is the distribution for the location of terrorist attacks in this region of the city.

\begin{figure}[ht!]
\centering
\includegraphics[width = .9\textwidth]{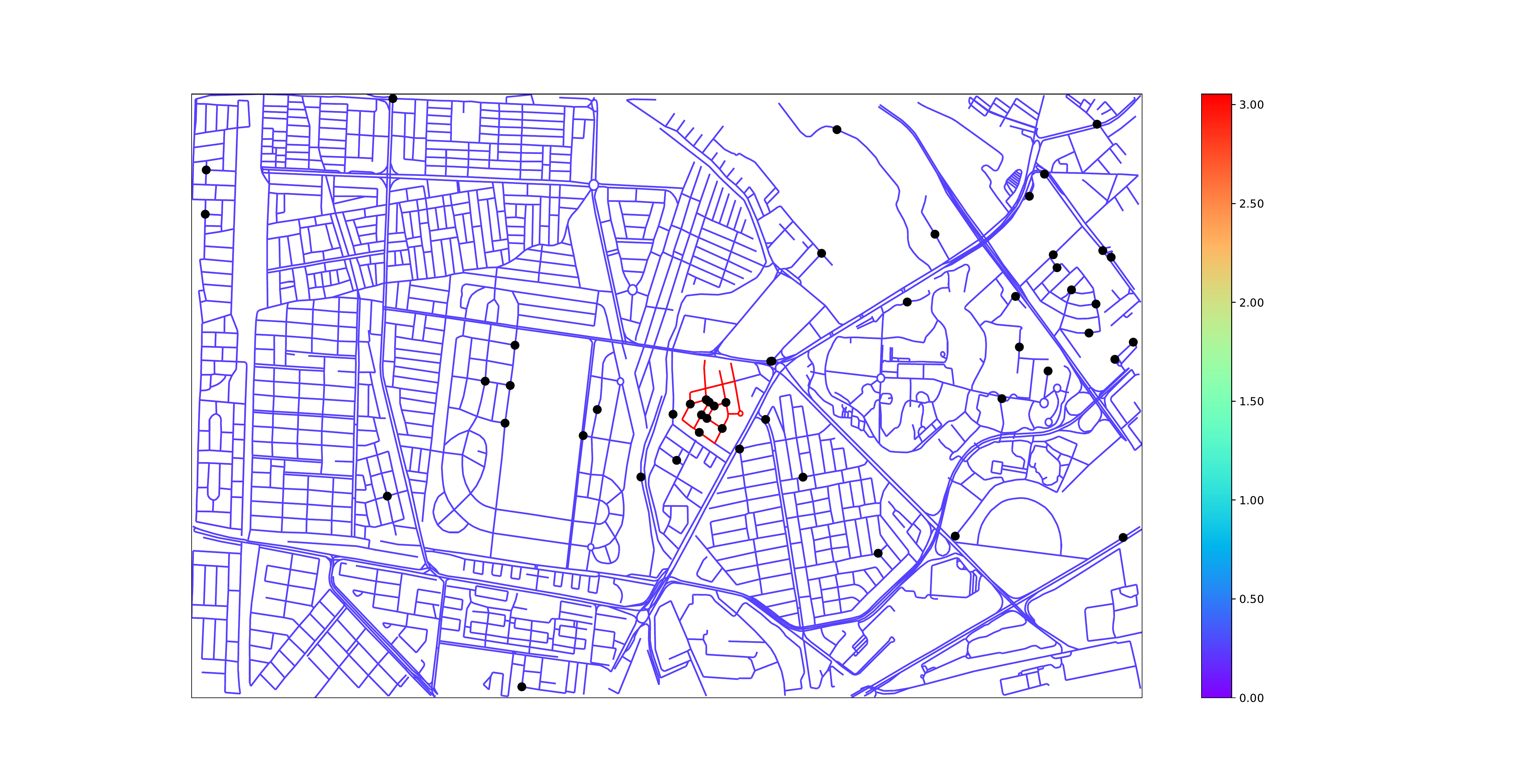}
\vspace{-.8cm}
\caption{An FDE for the location of terrorist attacks in a neighborhood of Baghdad. The detected hotspot contains the streets and alleys near a hospital.}
\label{fig:Baghdad}
\end{figure}

\begin{figure}[ht!]
\centering
\includegraphics[width = .9\textwidth]{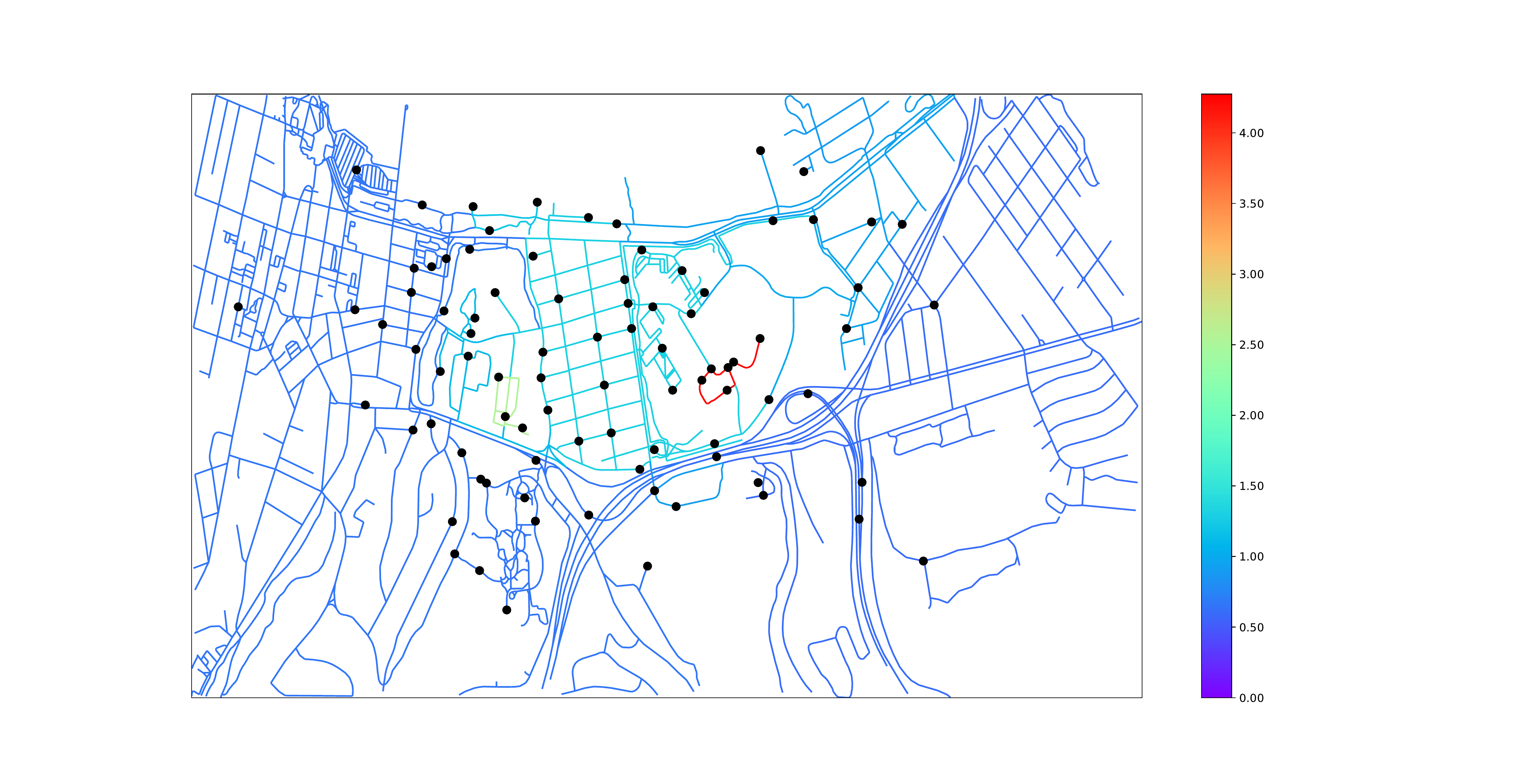}
\vspace{-.8cm}
\caption{A fused density estimator for artificial observations on Monterey's road network}
\label{fig:Monterey}
\end{figure}



Figure \ref{fig:Monterey} is an FDE on the road network in Monterey, California. The observations were generated according to a multivariate normal distribution, and projected onto the nearest waypoint in the OpenStreetMap dataset.
These examples of FDEs on geometric networks illustrate some important properties of the estimator. The FDEs clearly respect the network topology. This is most obviously demonstrated in the Monterey example, where the red and light green regions, which correspond to elevated portions of the density, are chosen to be sparsely connected regions of the network. This is intuitive because the sparsely connected regions impact the fusion penalty less severely than a highly connected region, but it is one way that FDEs reflect the underlying network structure. The Baghdad and Davis examples demonstrate that FDEs can also be used for hot spot localization, and especially in low-data circumstances. Lastly, we note that FDEs partition the geometric network into level sets, thereby forming various regions of the network into clusters. This clustering is an interesting aspect of FDEs, and suggests they could be used to classify regions into areas of high and low priority.

\subsection{Algorithmic Concerns}
The two most prevalent methods for solving sparse quadratic programs are interior point algorithms and the alternating direction method of multipliers. Interior point methods to solve problems of the form \eqref{PrimalProb} were introduced by \cite{Kim}. Interior point approaches have the benefit of requiring few iterations for convergence. The cost per iteration, however, depends crucially on the structure of $D_1$ and $D_2$ when performing a Newton step on the relaxed KKT system. In the case of univariate fused density estimators, the Newton step requires inversion of a banded matrix, one which has its nonzero elements concentrated along the diagonal. Leveraging the banded structure allows inversion to be performed in linear time, which is crucial to the performance of the algorithm. For further details of interior point methods, we refer the reader to \cite{WrightIP, ConvOpt, Nesterov}.

The alternating direction method of multipliers proceeds by forming an augmented lagrangian function and updating the primal and dual variables sequentially. More details can be found in \cite{ADMM, Bertsekas}. Compared to interior point methods, convergence of ADMM usually requires more iterations of a less-expensive update, whereas interior point methods converge in fewer iterations but require a more expensive update.
In this section, we compare the performance of these algorithms on fused density estimation problems. A comparison between the methods on the related problem of trend-filtering can be found in \cite{Sharpnack}, where the algorithmic preferences pertained only to the $2$x$2$ grid graph setting. Their results favor the ADMM approach, though the regularity of this graph structure makes generalizing to general graphs difficult.

For software, we use the Operator Splitting Quadratic Program (OSQP) solver and CVXOPT. These are mature sparse QP solvers that use ADMM and interior point algorithms, respectively. They are both open source, and compare favorably to commercial solvers \cite{OSQP, QPExper}. Our choice to use these solvers instead of custom implementations reflects that (i) these tools are representative of what is available in practice (ii) outsourcing this portion to other solvers reduces the ability for subtle differences in implementation to favor one method over the other (iii) these projects are production-quality, so their implementations are likely to be of higher quality than custom implementations.
We first compare ADMM and interior point methods on univariate fused density estimator problems. We perform 200 simulations, sampling 100 data points from each distribution. We let $\lambda$ range from $0.006$ to $0.1$. These choices correspond to the lower bound on the $\lambda$ parameter in Theorem \ref{RestateThm} and an upper bound which selects a constant or near-constant density. We report in-solver time, in seconds, and do not include the time required to convert to the sparse formats required for each solver. 

\begin{table}[ht]
\caption{Mean and standard deviation of run time (s) for univariate OSQP experiments}
\centering
\begin{tabular}{p{3cm}  ccc}
\toprule
& \multicolumn{3}{c}{$\lambda$}\\
\cmidrule{2-4} 
\textbf{Density} & $0.006$ & $0.05$ & $0.1$ \\
\cmidrule{1-4} Exponential & $0.0361 \pm 0.1310$ & $0.0051 \pm 0.0043$ & $0.0045 \pm 0.0045$ \\
Normal & $0.0209 \pm 0.0912$ & $0.0112 \pm 0.0569$ & $0.0052 \pm 0.0046$ \\
Uniform & $0.0269 \pm 0.1077$ & $0.0769 \pm 0.0565$ & $0.0074 \pm 0.0412$ \\ \bottomrule
\end{tabular}
\end{table}

\begin{table}[ht]
\caption{Mean and standard deviation of run time (s) for univariate CVXOPT experiments}
\centering
\begin{tabular}{p{3cm} ccc}
\toprule
& \multicolumn{3}{c}{$\lambda$}\\
\cmidrule{2-4} 
\textbf{Density} & $0.006$ & $0.05$ & $0.1$ \\
\cmidrule{1-4} 
Exponential & $0.0087 \pm 0.0012$ & $0.0069 \pm 0.0008$ & $0.0078 \pm 0.0011$ \\
Normal & $0.0086 \pm 0.0012$ & $0.0071 \pm 0.0009$ & $0.0076 \pm 0.0010$ \\
Uniform & $0.0087 \pm 0.0010$ & $0.0065 \pm 0.0008$ & $0.0061 \pm 0.0008$ \\
\bottomrule
\end{tabular}
\end{table}

In these experiments, interior point terminated in around $10$ iterations. The number of iterations in ADMM were less consistent, ranging from a few hundred to a few thousand.

For the geometric network case, we performed experiments using four examples: the San Diego and Baghdad datasets from figures \ref{fig:SanDiego} and \ref{fig:Baghdad}, in addition to similar datasets in Davis, California. One of these is a fused density estimator with domain as the road network in downtown Davis, and the other is on the \emph{entire city} of Davis--our largest example in this paper--which has 19000 variables and 25000 constraints in the dual formulation \eqref{DualProb}. We choose $\lambda$ in a range that progresses from overfitting to underfitting the data. By overfit, we mean that we choose $\lambda$ as small as possible to make the fused density estimator problem still feasible. By underfit, we mean that the fused density estimator is a constant function. We record `-' when a solver does not run to successful completion. All experiments were run on a computer with 8 GB of memory, an intel processor with four cores at 2.50 GHz, and a 64-bit linux operating system.

\begin{table}[ht]
\caption{OSQP run times (s) for geometric network examples}
\centering
\begin{tabular}{p{4cm} ccc}
\toprule
& \multicolumn{3}{c}{$\lambda$ parameter}\\
\cmidrule{2-4} 
\textbf{Example} & Overfit & Middle & Underfit \\
\cmidrule{1-4} Baghdad & $0.1086$ & $0.0686$ & $0.0639$ \\ 
San Diego & $0.0920$ & $0.0961$ & $0.0628$ \\ 
Downtown Davis & $0.0269$ & $0.0769$ & $0.0074$ \\ 
Davis & $12.0698$ & $0.8539$ & $0.6052$\\ 
\bottomrule
\end{tabular}
\end{table}

\begin{table}[ht]
\caption{CVXOPT run times (s) for geometric network examples}
\centering
\begin{tabular}{p{4cm} ccc}
\toprule
& \multicolumn{3}{c}{$\lambda$ parameter}\\
\cmidrule{2-4} 
\textbf{Example} & Overfit & Middle & Underfit \\
\cmidrule{1-4} Baghdad & $1.5493$ & $1.2813$ & $1.1268$ \\
San Diego & $0.5507$ & $0.4956$ & $0.3256$ \\
Downtown Davis & $0.0812$ & $0.5615$ & $0.4864$ \\
Davis & - & $13.4456$ & $13.3911$\\ 
\bottomrule
\end{tabular}
\end{table}

From these experiments we see that the augmented lagrangian method outperforms interior point on the geometric network examples. The lack of regularity in the matrices $D_1$ and $D_2$, and the large-scale matrix factorizations associated with Newton limits this method in comparison to ADMM. On smaller, well-structured problems, like in the univariate examples, interior point methods are often faster. On these well-structured problems, however, the gain in performance is negligible (on the order of a tenth of a second). On the other hand, the speed and versatility of OSQP, especially in the context of large, irregular network structure, leads us to recommend ADMM as the method to solve the fused density estimator problem in \eqref{DualProb}. This supports the suggestion of using ADMM for trend-filtering in \cite{Sharpnack}, and extends their recommendation beyond the $2 \times 2$ grid graph.

\section{Statistical Rates} \label{StatRates}

In this section we prove a squared Hellinger rate of convergence for fused density estimation when the true log-density is of bounded variation.
Hellinger distance is defined as
$$
h^2(f, f_0) = \frac{1}{2} \int_G (\sqrt f - \sqrt{f_0} )^2 \, dx,
$$
where $dx$ is the base measure over the edges in the geometric network $G$; in the univariate setting, this is just the Lebesgue measure. The factor of $\frac{1}{2}$ is a convention that ensures that the Hellinger distance is bounded above by $1$. The Hellinger distance is a natural choice for quantifying rates of convergence for density estimators because it is tractable for product measures and provides bounds for rates in other metrics \cite{LeCam73, LeCamBook, Gibbs}. The squared Hellinger risk of an estimator $\tilde{f}$ for $f_{0}$ is
$\mathbb{E}[h^2(\tilde{f}, f_{0})]$. The minimax squared Hellinger risk over a set of densities $\mathcal{H}$, for a sample size $n$ is
$$\min_{\tilde{f}} \max_{f \in \mathcal{H}} \mathbb{E}_{f}[h^2(\tilde{f}, f)].$$
The minimum is over all estimators $\tilde{f}$ which are measurable maps from the sample space of $x_1,...,x_n$ to $\mathcal{H}$.

We find fused density estimation achieves a rate of convergence in squared Hellinger risk which matches the minimax rate over all univariate densities in $\mathcal{F}$--densities of log-bounded variation where the underlying geometric network is simply a compact interval. In this sense, univariate FDE has the best possible squared Hellinger rate of convergence over this function class. The rate we attain is $n^{-2/3}$, and the equivalence of rates is asymptotic. On an arbitrary connected geometric network, minimax rates for density estimation can depend on the network, but our results demonstrate that FDE on a geometric network has squared Hellinger rate at most the univariate minimax rate. 

We begin by establishing the minimax rate over the class $\mathcal{F}$, which gives a lower bound on the squared Hellinger rate for fused density estimation. To establish the lower bound, it is sufficient to examine the minimax rate of convergence over a set of densities contained in $\mathcal{F}$. Fixing a constant $C$ and compact interval $I$, we consider the set of functions $g: I \to \R$
$$\BV(C) := \{g: \TV(g) \leq C, \norm{g}_{\infty} < C\}.$$ 
Recall that, for a given radius $\epsilon$, the packing entropy of a set $S$ with respect to a metric $d: S \times S \to \R_{+}$ is the logarithm of its packing number, the size of the largest collection of points in $S$ which are at least $\epsilon$-separated with respect to the metric $d$. Because $\BV(C)$ is bounded below, the packing entropy of $\BV(C)$ and $\widetilde{\BV}(C) := \{\exp(g) : g \in \BV(C)\}$ are of the same order. From Example 6.4 in \cite{YangBarron}, we have that $\BV(C)$ has $L_2$ packing entropy of order $\frac{1}{\epsilon}$.
Applying Theorem 5 from \cite{YangBarron} gives the minimax squared Hellinger rate over densities $\{\frac{f}{\int f} : f \in \widetilde{\BV}(C) \}$ as $n^{-2/3}$. In Theorem \ref{RateResult}, we show that the FDE attains the rate of $n^{-2/3}$ over the larger class $\cF$. Therefore, the minimax squared Hellinger rate over $\cF$ must also equal $n^{-2/3}$, so we have proven the following theorem. For sequences $a_{n}$ and $b_{n}$, we write $a_n \asymp b_n$ if $a_n = O(b_n)$ and $b_n = O(a_n)$.

\begin{theorem} \label{LowerBound} 
The minimax squared Hellinger rate over $\cF$, the set of densities $f$ with $\log f$ of bounded variation, is $n^{-2/3}$. That is,
$$\min_{\tilde{f}} \max_{f \in \mathcal{F}} \mathbb{E}_{f}[h^2(\tilde{f},f)] \asymp n^{-2/3}.$$
\end{theorem}

To prove the FDE rate of convergence for univariate density estimation, we extend techniques developed for the theory of M-estimators, \cite{vandeGeer}, and locally-adaptive regression splines in Gaussian models, \cite{MammenvandeGeer}.
A detailed proof of our main result can be found in Appendix \ref{AppA}.
This rate bound for FDE is based on novel empirical process bounds for log-densities of bounded variation, and these are used in conjunction with peeling arguments to provide a uniform bound on the Hellinger error.
The empirical process bounds in \ref{subsec2} rely on new Bernstein difference metric covering number bounds for functions of bounded variation, which can be found in Appendix \ref{AppB}.
We extend the FDE rates for the univariate setting to arbitrary geometric networks in section \ref{subsec3}; this requires embedding the geometric network onto the real line. This embedding is constructed from the depth-first search algorithm, a technique used in \cite{padilla} for regression over graphs, and is described in Appendix \ref{AppA}.

The subsections in this section follow this outline:
In subsection \ref{subsec1}, we provide a proof sketch of the squared Hellinger rate of convergence for the univariate FDE.
In subsection \ref{subsec2}, we detail the lemmas used to prove the main result.
In subsection \ref{subsec3}, we extend these rate results from the univariate setting to arbitrary geometric networks.

\subsection{Upper Bounds for Rate of Univariate FDE} \label{subsec1}

In this subsection we prove a squared Hellinger rate of $n^{-2/3}$ for univariate fused density estimation. Let the geometric network $G$ be a closed interval $[a,b]$ (a single edge connecting nodes $a$ and $b$). Recall the definition of $\mathcal{F}$ as the set of densities $f$ with $\log f$ of bounded variation. Let $f_{0} \in \mathcal{F}$ be a fixed density on $G$, so that the total variation   $\TV(\log f_0)$ is constant as $n$ increases.

\begin{theorem}\label{RateResult}
Let $\hat{f}_n$ be the fused density estimator of an iid sample of $n$ points drawn from a univariate density $f_0$.
There is an $f_0$-dependent sequence $\lambda_n$ such that $\lambda_n = O_P (n^{-2/3})$, the FDE is well defined, and
$$\mathbb{E}_{f_{0}}[h^{2}(\hat{f}_n,f_{0})] = O(n^{-2/3}).$$ 
Combined with the lower bound in Theorem \ref{LowerBound}, this gives that univariate fused density estimation attains the minimax rate over densities in $\mathcal{F}$.
\end{theorem}

\emph{Proof Sketch} (Detailed proof in Appendix \ref{AppA}).

In order to control the Hellinger error for FDE, we rely on the fact that the FDE is the minimizer of \eqref{fde}.
We derive an inequality involving the squared Hellinger distance, an empirical process, and fusion-penalty terms.
This inequality (and in general inequalities serving this purpose; see \cite{vandeGeer}) is referred to as a basic inequality.
To reduce notation, we introduce the shorthand $\hat{h} = h(\hat{f}_n, f_0)$, $I(f) = \TV(\log f)$, $\hat{I} = I(\hat{f}_n)$, $I_{0} = I(f_{0})$, and $p_f = \frac{1}{2} \log \frac{f + f_{0}}{2f_{0}}$.

We arrive at the following basic inequality by manipulating the optimality condition, $-P_{n}(\log \hat{f}_{n}) + \lambda_{n} \hat{I} \leq - P_{n}(\log f_{0}) + \lambda_{n} I_{0}$.
In fact, from the definition of the FDE we have the stronger condition $-P_{n}(\log \hat{f}_{n}) + \lambda_{n} \hat{I} \leq -P_{n}(\log f) + \lambda_{n} I(f)$ for all $f \in \mathcal{F}$, but the weaker condition will suffice.
\begin{lemma}[Basic Inequality]\label{SketchBasic}
$$\hat{h}^{2} \leq 16 (P_{n} - P)(p_{\hat{f}_{n}}) + 4 \lambda_{n}(I_{0} - \hat{I}).$$
\end{lemma}

Squared Hellinger rates now follow from controlling the right hand side. 
We do so by considering two cases. When $\hat{h}$ is small, we show that
$$(P_{n} - P)(p_{\hat{f}_{n}}) = O_{P}\left(n^{-2/3}(1+I_{0}+\hat{I})\right).$$
From the basic inequality, this gives
\begin{align} \label{SketchEq1}
\hat{h}^2 &= O_{P}\left(16 n^{-2/3}(1+I_{0}+\hat{I}) + 4 \lambda_{n}(I_{0}-\hat{I})\right) \nonumber \\
&= O_{P}\left(4(4 n^{-2/3}-\lambda_{n})\hat{I} + 4(4n^{-2/3} + \lambda_{n})I_{0} + 16 n^{-2/3}\right).
\end{align}
Excluding details, when $\lambda_n$ is chosen to dominate $4 n^{-2/3}$, the first term in \eqref{SketchEq1} is negative, so we conclude that $\hat{h}^2 = O_{P}\left(\max\{n^{-2/3},\lambda_{n}\}\right)$.

The condition ``when $\hat{h}$ is small'', and the corresponding control on $(P_{n} - P)(p_{\hat{f}_{n}})$ can be formalized in the following theorem.
\begin{theorem} \label{ThmSktch1}
$$\sup_{h(f,f_{0}) \leq n^{-1/3}(1+I(f)+I_{0})} \frac{n^{2/3} \abs{(P_{n}-P)(p_{f})}}{1+I(f)+I_{0}} = O_{P}(1),$$
where the supremum is taken over all $f \in \mathcal{F}$.
\end{theorem}

When $\hat{h}$ is large, on the other hand, we show that $(P_{n} - P)(p_{\hat{f}_{n}}) = O_{P}\left(n^{-1/2} \cdot \hat{h}^{1/2} \cdot (1 + \hat{I} + I_{0})^{1/2}\right)$. From the basic inequality, this gives
$$\sqrt{n} \hat{h}^2 = O_{P}\left(16 \hat{h}^{1/2}(1+I_{0}+\hat{I})^{1/2} + 4 \sqrt{n} \lambda_{n}(I_{0}-\hat{I})\right).$$
Whence we conclude that $\hat{h}^2 = O_{P}\left(\max\{n^{-2/3}, \lambda_{n}\}\right)$.
This follows from the analogue to \eqref{ThmSktch1} when $\hat h$ is large.
\begin{theorem} \label{ThmSktch2}
$$\sup_{h > n^{-1/3}(1+I(f)+I_{0})} \frac{n^{1/2} \abs{(P_{n}-P)(p_{f})}}{h^{1/2}(f,f_0)(1+I(f)+I_{0})^{1/2}} = O_{P}(1),$$
where the supremum is taken over all $f \in \mathcal{F}$.
\end{theorem}

To summarize our conclusions so far: the squared Hellinger rate is $\max\{n^{2/3}, \lambda_n\}$ when $\lambda_n$ balances the competing terms in \eqref{SketchEq1}. By choosing
$$\lambda_n = \max\left\{\sup_{h(f,f_{0}) \leq n^{-1/3}(1+I(f)+I_{0})} \frac{4 \abs{(P_{n}-P)(p_{f})}}{1+I(f)+I_{0}}, n^{-2/3}\right\},$$
we have a minimal $\lambda_n$ which dominates in \eqref{SketchEq1}. Furthermore, this choice of $\lambda_n$ satisfies $\lambda_n = O_{P}(n^{-2/3})$ by Theorem \ref{ThmSktch1}. We have established a squared Hellinger rate of $n^{-2/3}$ for both the cases of $\hat{h}$ considered. Furthermore, this choice of $\lambda_n$ satisfies the condition on $\lambda$ in Theorem \ref{RestateThm}, so the FDE is well-defined.

Theorems \ref{ThmSktch1} and \ref{ThmSktch2} are essential components of the proof outlined above. Both of these results are new and of independent interest. Their derivation requires the following lemma.
\begin{lemma} \label{LemSktch3}
Let $M \in \R$ and $\mathcal{P}_{M} = \{p_f: 1+I(f)+I_0 \leq M \}$. There is a constant $C$ and choice of $c_1$ such that for all $C_1 \geq c_1$ and $\delta \geq \frac{M}{2} \cdot n^{-1/3}$
$$\bP\left(\sup_{p_{f} \in \mathcal{P}_{M}, h(f,f_{0}) \leq \delta} \abs{\sqrt{n} (P_{n}-P)(p_{f})} \geq 2 C_{1} \sqrt{M} \delta^{1/2} \right) \leq C \exp\left[- \frac{C_{1} M \delta^{-1}}{4 C^2}\right]$$
\end{lemma}
Lemma \ref{LemSktch3} can be used to prove Theorems \ref{ThmSktch1} and \ref{ThmSktch2} by applying the peeling device twice, once each for the parameters $M$ and $\delta$.

The proof of Lemma \ref{LemSktch3} requires three basic ingredients: control of the bracketing entropy of $\mathcal{P}_M$, a uniform bound on $\mathcal{P}_M$, and a relationship dictating how $M$ scales with control of the Hellinger distance. These ingredients have the same motivation as in \cite{MammenvandeGeer}, where the authors use a total variation penalty to construct adaptive estimators in the context of regression. In that work, the authors assume subgaussian errors and prove bounds on metric entropy for functions of bounded variation. The subgaussian assumption provides local error bounds, and the metric entropy condition bounds the number of sets on which we must control that error. Though similarly motivated, our context is more complicated. In order to control the $M$ in $\mathcal{P}_{M}$ with the Hellinger distance, we consider coverings in the Bernstein difference metric instead of the $L_{2}(P)$ metric. Using the Bernstein difference allows us to achieve the results in Lemma \ref{LemSktch3}, but its use requires control of generalized bracketing entropy--bracketing with the Bernstein difference--instead of the usual bracketing entropy with the $L_{2}(P)$ metric. In addition, the uniform bound we require is now on the Bernstein difference over $\mathcal{P}_{M}$.

In Appendix \ref{AppB}, we show that the bracketing entropy of $\mathcal{P}_M$, with bracketing radius $\delta$, is of order $\frac{M}{\delta}$. This bracketing entropy results implies generalized bracketing entropy bounds, and can be proved similarly to results in monotonic shape-constrained estimation \cite{vandervaart}. In order to achieve the finite sample bounds necessary to achieve these rates, Bernstein's inequality is used to provide concentration inequalities that are critical to bounding the basic inequality.
With this combination of local error bounds and bracketing rates, we can apply results in the spirit of generic chaining \cite{talagrand} to obtain Lemma \ref{LemSktch3}.

Lastly, we translate the probabilistic results into bounds on Hellinger risk. In general, one cannot prove expected risk rates from convergence in probability because the tails may not decay quickly enough to give a finite expectation. But out of the proofs of Theorems \ref{ThmSktch1} and \ref{ThmSktch2}, we can derive exponential tail bounds for $h^2(\hat{f}_{n}, f_0)$. This allows us to translate our probabilistic rates into rates on the Hellinger risk; doing so requires some care to simultaneously apply the rates in Theorems \ref{ThmSktch1} and \ref{ThmSktch2}. These details are provided in the expanded proof in Appendix \ref{AppA}.

\subsection{Guarantees for Connected Geometric Networks}\label{subsec3}

In the previous sections we proved an $n^{-2/3}$ rate of convergence for univariate fused density estimators. The following theorem extends this result to arbitrary geometric networks.

\begin{theorem}\label{NetRate}
Let $\hat f_n$ be the FDE of an iid sample over a connected geometric network with true density $f_0$.
Then there exists a choice of $\lambda_n$, dependent on $f_0$, such that $\lambda_n = O_P (n^{-2/3})$ and 
$$\mathbb{E}_{f_0}\left[ h^2(\hat f_n, f_0) \right] = O(n^{-2/3}).$$
\end{theorem}

We prove this theorem using an embedding lemma, Lemma \ref{TVBounds}, which states that for any fixed geometric network $G$ there is a measure-preserving embedding $\gamma$ of $G$ into $\R$ that preserves densities and Hellinger distances. Furthermore, for any function $g$ on $G$, the (univariate) total-variation of the embedded function $g \circ \gamma^{-1}$ never exceeds twice that of the graph-valued total variation. With this lemma in hand, Theorem \ref{NetRate} is proven by strategically bounding terms in our analysis by their univariate counterparts. A detailed proof can be found in the supplementary material.

Theorem \ref{NetRate} provides an upper-bound on the minimax Hellinger rate for densities of log-bounded variation on geometric networks. Unlike the unvariate case, however, we do not have a lower bound as in Theorem \ref{LowerBound}, so we cannot conclude that FDE attains the minimax rate for arbitrary geometric networks. This mirrors similar results found in the Gaussian regression setting; see for example \cite{padilla}. The minimax squared-Hellinger rate for density estimation on geometric networks is at least as small as the univariate rate, and it is reasonable to suspect that the minimax rate on some graphs may be strictly better than the univariate rate. Though the univariate case may seem simple, and hence one might expect it to be easier, the sparse connectivity of the underlying graph in univariate estimation negatively affects its minimax rate of convergence. In some sense this is intuitive. 
For example, adding cycles to a graph increases the total variation compared the same graph with the cycles removed.
The increased total variation can be seen as applying more shrinkage in the context of estimation, which makes total variation balls smaller and the problem easier. 
Similarly, tree graphs (graphs without cycles) have more connectivity than the univariate chain graph and larger total variation for a function defined on it. 
This intuition is consistent with the formal results from \cite{OptRatesforDenoising} in the regression setting. 
While there may be networks for which the FDE and minimax squared Hellinger rates may decrease more quickly than the $n^{-2/3}$, we leave that study to future work.

\subsection*{Acknowledgements}
RB was supported in part by the U.S. Office of Naval Research grant N00014-17-2372. JS was supported in part by NSF grant DMS-1712996. We would like to thank Ryan Tibshirani, Roger Wets, and Matthias K{\"o}ppe for helpful conversations.

{\small
\printbibliography
}

\newpage


\appendix
\section{Appendix: Proofs} \label{AppA}

\subsection{Proofs from Section \ref{Computation}}

\emph{Proof of Lemma \ref{UnCon}.}



Let $\hat{g}$ be any function of bounded variation. Set $\bar{g} = \hat{g} - \log\left(\int_{G} e^{\hat{g}} \, dx \right)$, so that $\int e^{\bar{g}} \, dx = 1$. The value of 
\begin{equation}\label{UnConExp1}
-P_{n}(g) + \lambda \TV(g) + \int_{G} e^g \, dx
\end{equation}
evaluated at $\hat{g}$ is
\begin{equation}\label{UnConExp2}
-P_{n}(\hat{g}) + \lambda \TV(\hat{g}) + \int_{G} e^{\hat{g}} \, dx.
\end{equation}
Whereas \eqref{UnConExp1} evaluated at $\bar{g}$ is
\begin{equation}\label{UnConExp3}
-P_{n}(\hat{g}) + \log\left(\int_{G} e^{\hat{g}} \,dx \right) + \TV(\hat{g}) + 1.
\end{equation}
Here we have used that $\TV(\hat{g}) = \TV(\bar{g})$, which follows from shift-invariance of total variation. That is, $\TV(g) = \TV(g+a)$ for any function $g$ and constant $a$.
Subtracting \eqref{UnConExp3} from \eqref{UnConExp2} gives
$$\int_{G} e^{\hat{g}} \, dx - \log(\int_{G} e^{\hat{g}} \, dx) - 1.$$
Recall that $x - \log(x) \geq 1$ for all $x>0$, with equality attained if and only if $x = 1$. This proves our result, since we have shown that $\hat{g}$ cannot be a minimizer of \eqref{UnConExp1} unless $\hat{g} = \bar{g}$, in which case $\int_{G} e^{\hat{g}} \, dx = 1$.
\qed

In is worth noting that, in the proof of Lemma \ref{UnCon}, we only require that total variation is shift invariant. A similar result holds with any shift invariant penalty substituted for $\TV$.

\emph{Proof of Representer Theorem, Lemma \ref{FinDim}.}

Let $e \in E$, and assume that $e$ is identified with some interval $[0,l]$. Consider any subinterval $(a,b)$ of $e$ which does not intersect $\{x_{1},...,x_{n}\} \cup V$. Assume, towards a contradiction, that $\hat{g} = \log \hat{f}$ is a function which is not constant on $(a,b)$. We will show that $\hat{g}$ cannot minimize \eqref{UnConProb}.

Define 
$$\bar{g}_{e}(x) = \begin{cases} \min\{\hat{g}(a), \hat{g}(b)\} & \text{for } x \in (a,b) \\ \hat{g}(x) & \text{otherwise}. \end{cases}$$
Let $\bar{g}$ be a function on $G$ which is $\bar{g}_{e}$ on edge $e$ and $\hat{g}_{e}$ otherwise.
We next consider the effect of this change on the objective in \eqref{UnConProb}. Since no $x_{i} \in (a,b)$, the $P_{n}$ term is unaltered by changing $\hat{g}$ to $\bar{g}$. The interval $(a,b)$ is contained in $e$, so we have
$$\TV(g_{e}) = \TV(g_e|_{[0,a]}) + \TV(g_{e}|_{[a,b]}) + \TV(g_e|_{[b,L]})$$
for every real-valued function $g$ on $G$. From the definitions of $\hat{g}$ and $\bar{g}$,
$$\TV(\bar{g}_{e}|_{[a,b]}) = \abs{\hat{g}(a) - \hat{g}(b)} \leq \TV(\hat{g}_{e}|_{[a,b]}).$$
This equality is attained if and only if $\hat{g}$ is monotonic on $[a,b]$. If follows that $\TV(\bar{g}) \leq \TV(\hat{g})$.

The integral term in \eqref{UnConProb} is less at $\bar{g}$ than its evaluation at $\hat{g}$, since $\bar{g} \leq \hat{g}$. Equality holds when $\{x: \hat{g}_{e}(x) \neq \bar{g}_{e}(x)\}$ has measure zero on $(a,b)$. Hence,
$$- P_{n}(\hat{g}) + \lambda \TV(\hat{g}) + \int_{G} e^{\hat{g}} \, dx \leq - P_{n}(\bar{g}) + \lambda \TV(\bar{g}) + \int_{G} e^{\bar{g}} \, dx.$$
We cannot have that $\hat{g}_{e}|_{[a,b]}$ is monotonic and satisfies $\{x \in (a,b): \hat{g}_{e}(x) \neq \bar{g}_{e}(x)\}$ has measure zero, unless $\hat{g}_{e}$ is constant on $(a,b)$. By assumption it is not, so we conclude that $\hat{g}$ cannot minimize \eqref{UnConProb} since its evaluation at the objective is strictly greater than at $\bar{g}$. Therefore, any $\hat{g}$ that satisfies \eqref{UnConProb} must be constant on $(a,b)$.
\qed

\emph{Proof of Theorem \ref{RestateThm}.}


We pick up from the discussion preceding the theorem's statement. Recall that the subdifferentials of the fusion penalty terms are preserved by the exponential transformation. This allows to conclude that there is a $\hat{c}, \hat{k}$ satisfying
\begin{align*} 
0 \in & \partial \left( \sum_{e \in E} \left\{ -\frac{1}{2n}\sum_{i=1}^{n_{e}} (c_{e,i}+c_{e,i+1}) + \left(\lambda-\frac{1}{2n}\right)\sum_{i=1}^{n_{e}} \abs{c_{e,i}-c_{e,i+1}} + \sum_{i=1}^{n_{e}+1} s_{e,i} e^{c_{e,i}} \right\} \right. \\
&+\left. \lambda \sum_{v \in V} \sum_{e \in \inc(v)} \abs{k_{v}-c_{e,v}}\right)_{(\hat{c}, \hat{k})},
\end{align*}
if and only if $\hat{z} = e^{\hat{c}}$ and $\hat{h} = e^{\hat{k}}$ satisfies
\begin{align} \label{NewOptCond}
0 \in & \partial \left( \sum_{e \in E} \left\{ -\frac{1}{2n}\sum_{i=1}^{n_{e}} (z_{e,i}+z_{e,i+1}) + \left(\lambda-\frac{1}{2n}\right)\sum_{i=1}^{n_{e}} \abs{z_{e,i}-z_{e,i+1}} + \sum_{i=1}^{n_{e}+1} \frac{s_{e,i}}{2} \cdot z_{e,i}^2  \right\} \right. \\
&+ \left. \lambda \sum_{v \in V} \sum_{e \in \inc(v)} \abs{h_{v}-z_{e,v}}\right)_{(\hat{z}, \hat{h})}. \nonumber
\end{align}
The above subdifferentials are taken with respect to $(c,k)$ and $(z,h)$, respectfully. The problem that generates the optimality condition \eqref{NewOptCond} is
$$\min_{z,h} \sum_{e \in E} \left\{ -\frac{1}{2n}\sum_{i=1}^{n_{e}} (z_{e,i}+z_{e,i+1}) + \left(\lambda-\frac{1}{2n}\right)\sum_{i=1}^{n_{e}} \abs{z_{e,i}-z_{e,i+1}} + \sum_{i=1}^{n_{e}+1} \frac{s_{e,i}}{2} \cdot z_{e,i}^2  \right\} + \lambda \sum_{v \in V} \sum_{e \in \inc(v)} \abs{h_{v}-z_{e,v}}.$$
By solving this new problem, and then applying a log-transformation, we solve the original FDE problem. Recall that the original formulation of the problem \eqref{fde} was formulated in terms of the log-density $g$. Hence a solution to \eqref{NewOptCond} gives the values of density instead of the log-density.

By construction $w$ satisfies
$$w_{e,i} = \begin{cases} -\frac{1}{2n} & i = 1  \text{ or } i=n_{e}\\ -\frac{1}{n} & \text{ otherwise.}\end{cases}$$
By construction, we also have that
$$
\norm{D_1 z + D_2 h}_{1} = \left( \lambda - \frac{1}{2n} \right)\sum_{e \in E} \sum_{i=1}^{n_{e}} \abs{z_{e,i}-z_{e,i+1}} + \lambda \sum_{v \in V} \sum_{e \in \inc(v)} \abs{h_{v}-z_{e,v}}.
$$

Letting $S = \diag(s)$, we conclude that we can solve the fused density estimator problem by solving
$$\min_{z,h} \frac{1}{2} z^{\top} S z + w^{\top}z + \norm{ D_1 z + D_2 h}_{1}$$
because we have translated the optimality conditions in \eqref{NewOptCond} to the problem above. The FDE $\hat{f}$ is found by taking the piecewise constant portion of $\hat{f}_{e,i}$ to be $\hat{z}_{e,i}$ and the value of $\hat{f}$ at node $v$ to be $\hat{h}_{v}$.
\qed

\begin{theorem}[Extension of Theorem \ref{RestateThm}] \label{ExtendThm}
Let $x_1,...,x_n$ be the distinct locations of observations on a geometric network $G$. Partition these locations into the edges they occur on and the order in which they occur, so that $x_{e,i}$ denotes the $i$th observation along edge $e$.
\begin{itemize}
\item Let $q_{e,i}$ denote the number of observations which occur at location $x_{e,i}$.
\item Let $\deg(x_{e,i})$ denote the number of edge segments incident to the observation $x_{e,i}$. That is, $\deg(x_{e,i}) = 2$ is $x_{e,i}$ is in the interior of an edge and $\deg(x_{e,i})$ is the degree of the node in the graph when $x_{e,i}$ occurs at a node.
\item Let $z$ be a vector with indices enumerating the constant portions of the fused density estimator $\hat{f}$, such that $z_{e,i}$ denotes the value of the fused density estimator on the open interval between $x_{e,i}$ and $x_{e,i-1}$, or between an observation and the end of the edge if $i=1$ or $n_{e}+1$. 
\item Let $s_{e,i}$ be the length of the segment that determines $z_{e,i}$ and $S = \diag(s)$.
\item Let $h$ be a vector with indices enumerating the nodes in $G$, such that $h_{v}$ denotes the value of the fused density estimator at node $v$. 
\item Using the convention that $q_{e,0} = 0$ and $q_{e,n_{e+1}} = 0$. Define $\bar{q}$ such that $\bar{q}_{e,i} = \frac{q_{e,i}+q_{e,i-1}}{2}$, for each $e \in E$ and $i \in \{1,...,n_{e}+1\}$.
\item Let $r$ be a vector whose indices enumerate the vertices of $G$, such that $r_v$ denotes the number of observations that occur at node $v$. 
\item Let $C_1$ and $C_2$ be as in Theorem \ref{RestateThm}. That is, $C_1$ and $C_2$ are matrices with $n_{1}+n_{2}$ rows and elements in $\{-1,0,1\}$. We have that $\TV(f) = \norm{C_1 z + C_2 h}_{1}$, and $C_2$ is identically zero on its first $n_1$ rows while having a nonzero element in each of the remaining rows. Let 
$$B = \left(\begin{array}{cc} \diag(\lambda-q/2n) & 0_{n_1 \times n_2} \\ 0_{n_2 \times n_1} & \lambda I_{n_{2} \times n_{2}} \end{array} \right).$$
\item Let $D_1$ and $D_2$ denote the matrices $B C_1$ and $B C_2$, respectfully.
\item Lastly, let $u = -r/n$ and $w = -\bar{q}/n$.
\end{itemize}
Assume the penalty parameter $\lambda$ satisfies $\lambda > \max_{e,i} \left\{\frac{q_{e,i}}{n \cdot \deg(x_{e,i})} \right\}$. Then one can compute the fused density estimator $\hat{f}$ for this sample by solving
\begin{equation} \label{PrimalProb2}
\min_{z,h} \; \frac{1}{2} z^\top S z + w^\top z  + u^\top h + \norm{D_1 z + D_2 h}_{1}.
\end{equation}
\end{theorem}

\begin{proof} The proof of this theorem follows exactly as in the proof of Theorem \ref{RestateThm}, with slightly more cumbersome notation.
\end{proof}

The following is a more general statement of Proposition \ref{Dual}, and provides the dual of the more general primal problem, \eqref{PrimalProb2}. 

\begin{proposition}[Extension of Proposition \ref{Dual}]\label{Dual2}
The dual problem to \eqref{PrimalProb2} is 
\begin{align}\label{DualProb2}
\min_{y} & \quad \frac{1}{2} y^\top D_{1} S^{-1} D_{1}^\top y + w^\top S^{-1} D_{1}^\top y \nonumber\\
& \quad \norm{y}_{\infty} \leq 1 \\
& \quad D_{2}^\top y = -u. \nonumber
\end{align} 
The primal solution $\hat{z}$ can be recovered from the dual $\hat{y}$ through the expression
$$\hat{z} = -S^{-1}(D_{1}^\top \hat{y}+w).$$
\end{proposition}

\begin{proof}
Write \eqref{PrimalProb} as
\begin{align*}
\min_{z,h,l} & \quad \frac{1}{2} z^\top S z + w^\top z + u^\top h + \norm{l}_{1} \\
& \st \quad l = D_{1}z + D_{2} h
\end{align*}
Introducing the dual variable $y$, this problem has Lagrangian
\begin{equation}\label{lagrangian}
\frac{1}{2}z^\top S z + w^\top z + u^\top h + \norm{l}_{1} + y^\top(D_{1} z + D_{2} h - l).
\end{equation} 
To find the dual problem, we minimize in the primal variables. This gives
\begin{equation} \label{label1}
\min_{l} -y^\top l + \norm{l}_{1} = \begin{cases} 0 & \text{ if } \norm{y}_{\infty} \leq 1 \\ -\infty & \text{ otherwise}. \end{cases}.
\end{equation}
In addition, we have the terms
\begin{equation}\label{label2}
\min_{z} \; \frac{1}{2} z^\top S z + w^\top z + y^\top D_{1} z
\end{equation}
and
\begin{equation} \label{label3}
\min_{h} \; u^\top h + y^\top D_{2} h.
\end{equation}
For \eqref{label2}, we have the optimality condition
\begin{equation}\label{OptCond1}
S z + w + D_{1}^\top y = 0.
\end{equation}
For \eqref{label3}, we require $D_{2}^\top y = -u$. Substituting \eqref{label1}-\eqref{label3} into \eqref{lagrangian}, we arrive at the dual problem
\begin{align*}
\max_{y} & \quad -\frac{1}{2} y^\top D_{1} S^{-1} D_{1}^\top y - w^\top S^{-1} D_{1}^\top y \\
& \quad \norm{y}_{\infty} \leq 1 \\
& \quad D_{2}^\top y = -u
\end{align*}
Translating this maximum into a minimum, and using the optimality condition in \eqref{OptCond1}, we have the result.
\end{proof}

\emph{Proof of Proposition \ref{OrderingProp}.}

We will prove that $s_{e,i} \leq s_{e,i+1}$ implies $z_{e,i} \geq z_{e,i+1}$. The second claim follows symmetrically. Assume, for contradiction, that $s_{e,i} \leq s_{e,i+1}$ and $\hat{z}_{e,i} < \hat{z}_{e,i+1}$. 

The condition for optimality in \eqref{PrimalProb} is
$$\left. 0 \in \partial \left( \frac{1}{2} z^\top S z + w^\top z + \norm{ D_1 z + D_2 h}_{1} \right)\right|_{\hat{z}, \hat{h}}.$$
The value of the subdifferential in the index corresponding to $z_{e,i}$ is
$$\partial \left(\frac{1}{2} s_{e,i} z_{e,i}^{2} +\frac{1}{n} z_{e,i} + \left(\lambda -\frac{1}{2n}\right)\left(\abs{z_{e,i-1}-z_{e,i}}+\abs{z_{e,i}-z_{e,i+1}}\right) \right).$$
Under the assumption that $\hat{z}_{e,i} < \hat{z}_{e,i+1}$, its evaluation at $\hat{z}$ is
$$\hat{z}_{e,i} s_{e,i} + \frac{1}{n} - \left(\lambda-\frac{1}{2n}\right) + (\lambda- \frac{1}{2n}) \partial(\abs{z_{e,i-1}-z_{e,i}})|_{\hat{z}}.$$
Similarly, the $e,i+1$ index evaluates to
$$\hat{z}_{e,i+1} s_{e,i+1} + \frac{1}{n} + \left(\lambda-\frac{1}{2n}\right) + (\lambda- \frac{1}{2n}) \partial(\abs{z_{e,i+2}-z_{e,i+1}})|_{\hat{z}}.$$
Since $-1 \leq \partial \abs{\cdot} \leq 1$, we have that
\begin{align*}
\hat{z}_{e,i} s_{e,i} + \frac{1}{n} - 2\left(\lambda-\frac{1}{2n}\right) &\leq  \hat{z}_{e,i} s_{e,i} + \frac{1}{n} - \left(\lambda-\frac{1}{2n}\right) + (\lambda- \frac{1}{2n}) \partial(\abs{z_{e,i-1}-z_{e,i}})|_{\hat{z}} \\ 
&\leq \hat{z}_{e,i} s_{e,i} + \frac{1}{n}
\end{align*}
and 
\begin{align*}
\hat{z}_{e,i+1} s_{e,i+1} + \frac{1}{n} &\leq  \hat{z}_{e,i+1} s_{e,i+1} + \frac{1}{n} + \left(\lambda-\frac{1}{2n}\right) + (\lambda- \frac{1}{2n}) \partial(\abs{z_{e,i}-z_{e,i+1}})|_{\hat{z}} \\ 
&\leq \hat{z}_{e,i+1} s_{e,i+1} + \frac{1}{n} + 2\left(\lambda - \frac{1}{2n}\right).
\end{align*}
Under the assumption that $\hat{z}$ solves this problem, we have that $0$ is in the $(e,i)$ index of the subdifferential. This implies
$$\hat{z}_{e,i} s_{e,i} + \frac{1}{n} - 2\left(\lambda-\frac{1}{2n}\right) \leq 0 \leq \hat{z}_{e,i} s_{e,i} + \frac{1}{n}.$$
But this inequality gives that $0$ is not in the $(e,i+1)$ index of the subdifferential, since
$$\hat{z}_{e,i} s_{e,i} + \frac{1}{n} < \hat{z}_{e,i+1} s_{e,i+1} + \frac{1}{n}.$$
This contradicts $\hat{z}$ as solving \eqref{PrimalProb}, so the result is proven.
\qed

\subsection{Proofs from Section \ref{StatRates}}
\emph{Proof of Theorem \ref{RateResult}.} 

We first show that $\hat{h}^2 = O_{P}(n^{-2/3})$. Fixing $\epsilon > 0$, we want to show there are $M \in \R$ and $N \in \mathbb{N}$ such that $n \geq N$ gives $\bP(n^{2/3} \hat{h}^2 > M) < \epsilon$.

We will show momentarily that $\hat{h}^2 = O_{P}(n^{-2/3})$ in both the cases when $\hat{h} \leq n^{-1/3}(1+\hat{I} + I_{0})$ and $\hat{h} > n^{-1/3}(1+\hat{I}+I_{0})$. Once we have established that both cases are $O_{P}(n^{-2/3})$, there exists $M_{1}, M_{2} \in \R$, $N_{1}, N_{2} \in \N$ such that $n \geq N_{1}$ gives
$$\bP\left(\left\{n^{2/3} \hat{h}^2  \geq M_{1} \right\} \bigcap \left\{\hat{h} \leq n^{-1/3}(1+\hat{I}+I_{0})\right\}\right) < \epsilon/2$$
and $n \geq N_{2}$ gives
$$\bP\left(\left\{ n^{2/3} \hat{h}^2 \geq M_{2} \right\} \bigcap \left\{\hat{h} > n^{-1/3}(1+\hat{I}+I_{0})\right\}\right) < \epsilon/2.$$
Therefore, for $N = \max\{N_{1}, N_{2}\}$ and $M = \max\{M_{1}, M_{2}\}$,
\begin{align*}
& \bP\left(n^{2/3} \hat{h}^2 > M \right) \\
=&\bP\left(\left\{n^{2/3} \hat{h}^2 \geq M\right\} \bigcap \left\{\hat{h} \leq n^{-1/3}(1+\hat{I}+I_{0})\right\}\right) \\
+& \bP\left(\left\{n^{2/3} \hat{h}^2 \geq M \right\} \bigcap \left\{\hat{h} > n^{-1/3}(1+\hat{I}+I_{0})\right\}\right) \\
<& \epsilon/2 + \epsilon/2 = \epsilon.
\end{align*}
This gives that $\hat{h}^2 = O_{P}(n^{-2/3})$.

We turn next to showing that $\hat{h}^2 = O_{P}(n^{-2/3})$ in both of the cases indicated. From the basic inequality, Lemma \ref{SketchBasic}, we have
$$\hat{h}^{2} \leq 16 (P_{n} - P)(p_{\hat{f}_{n}}) + 4 \lambda_{n}(I_{0} - \hat{I}).$$

Take 
$$\lambda_{n} = \max\left\{\sup_{h(f, f_{0}) \leq n^{-1/3}(1+I(f)+I_{0})} \frac{4\abs{(P_{n} - P)(p_{f})}}{1+I(f)+I_{0}}, n^{-2/3}\right\}.$$
The maximum guarantees that $\lambda_n$ satisfies the assumption on $\lambda$ in Theorem \ref{RestateThm} for $n$ large enough, so that $\hat{f}_{n}$ is well-defined. 

We prove in Theorems \ref{InCor} and \ref{OutThm} that
\begin{equation} \label{equate1}
\sup_{h(f, f_{0}) > n^{-1/3}(1+I(f)+I_{0})} \frac{n^{1/2} \abs{(P_{n} - P)(p_{f})}}{h^{1/2}(f, f_{0})(1+I(f)+I_{0})^{1/2}} = O_{P}(1) 
\end{equation}
and
\begin{equation} \label{equate2}
\sup_{h(f, f_{0}) \leq n^{-1/3}(1+I(f)+I_{0})} \frac{n^{2/3} \abs{(P_{n} - P)(p_{f})}}{1+I(f)+I_{0}} = O_{P}(1).
\end{equation}
Equation \eqref{equate2} gives that $\lambda_{n} = O_{P}(n^{-2/3})$.

First, assume that $h(\hat{f}_{n}, f_{0}) \leq n^{-1/3}(1+I+I_{0})$,
\begin{align} 
\hat{h}^2 &\leq 16 (P_{n} - P)(p_{\hat{f}_{n}}) + 4 \lambda_{n}(I_{0} - \hat{I}) \label{6}\\
&= (P_{n} - P)(p_{\hat{f}_{n}}) + 4 \lambda_{n}(1 + 2 I_{0}) - 4 \lambda_{n}(1+I_{0}+\hat{I})\\
&= 4 (1 + I_{0} + \hat{I})\left( \frac{4(P_{n} -P)(p_{\hat{f}_{n}})}{1 + I_{0} + \hat{I}} - \lambda_{n}\right) + 4 \lambda_{n}(1+ 2 I_{0})\\
&\leq 4(1+I_{0}+\hat{I}) \left(\sup_{h(f,f_0) \leq n^{-1/3}(1+I(f)+I_{0})} \frac{4(P_{n} - P)(p_{f})}{1+I_{0}+I(f)} - \lambda_{n} \right) + 4 \lambda_{n}(1+2 I_{0}) \label{3}
\end{align}

Our choice of $\lambda_{n}$ gives that the left term in this expression is less than or equal to zero. 
We conclude that
\begin{equation}\label{RandomEquation}
\hat{h}^2 \leq 4 \lambda_n (1+2 I_{0}).
\end{equation}
And finally
$$\frac{\hat{h}}{\sqrt{\lambda_n}} \leq 2 \sqrt{1+2 I_{0}}.$$
By our choice of $\lambda_n$, this bound gives that $\hat{h}^2 = O_{P}(n^{-2/3})$.

Assume next that $\hat{h} > n^{-1/3} (1 + \hat{I} + I_{0})$. Define subsets of the probability space
\begin{equation}\label{B_L}
B_L = \left\{\sqrt{n} \abs{(P_{n} - P)(p_{\hat f})} > L \cdot \hat{h}^{1/2} \cdot (1+\hat{I}+I_{0})^{1/2} \right\}
\end{equation}
and
\begin{equation}\label{C_M}
C_{M} = \{\hat{I} > M \geq I_{0}\}.
\end{equation}
By \eqref{equate1}, for each $\epsilon$ there is a corresponding $L$ such that $\bP(B_{L}) < \epsilon$.

On 
$B_{L}^{c} \cap C_{M}$, $(I_{0} -\hat{I}) < 0$, so from \eqref{SketchBasic}
$$\sqrt{n} \hat{h}^2 \leq 16 \cdot L \cdot \hat{h}^{1/2} (1+\hat{I} + I_{0})^{1/2} + 4 \lambda_{n} \sqrt{n} (I_{0} - \hat{I}) \leq 16 \cdot L \cdot \hat{h}^{1/2} (1+ \hat{I} + I_{0})^{1/2}.$$
Therefore,
$$\hat{h}^{1/2} \leq n^{-1/6} \left(16 \cdot L \cdot (1 + \hat{I} + I_{0})^{1/2} \right)^{1/3}.$$
Again using \eqref{SketchBasic}, we have
\begin{align*}
\sqrt{n} \hat{h}^{3/2} \leq 16 \cdot L \cdot (1+\hat{I} + I_{0})^{1/2} + \frac{4 \lambda_{n} \sqrt{n}(I_{0} - \hat{I})}{\hat{h}^{1/2}} \\
\leq 16 \cdot L \cdot (1+\hat{I} + I_{0})^{1/2} + \frac{4 \lambda_{n} \sqrt{n} (I_{0} - \hat{I})}{n^{-1/6} \left( 16 \cdot L \cdot (1 + \hat{I} + I_{0})^{1/2} \right)^{1/3}} \\
= 16 \cdot L \cdot (1+\hat{I} + I_{0})^{1/2} + \frac{4 \lambda_{n} n^{2/3}(I_{0} - \hat{I})}{\left( 16 \cdot L \cdot (1 + \hat{I} + I_{0})^{1/2} \right)^{1/3}}.
\end{align*}
The second inequality follows because $I_{0} - \hat{I} < 0$ on $C_M$. The definition of $\lambda_n$ gives that $\lambda_{n} n^{2/3} \geq 1$. Hence,
$$\leq 16 \cdot L \cdot (1+\hat{I} + I_{0})^{1/2} + \frac{4 (I_{0} - \hat{I})}{\left( 16 \cdot L \cdot (1 + \hat{I} + I_{0})^{1/2} \right)^{1/3}}.$$

The order of the left term is $\sqrt{\hat{I}}$, whereas the order of the right is $\hat{I}^{5/6}$. This gives that for $M$ large enough $\sqrt{n} \hat{h}^{3/2} < 0$. Of course this is not possible, so we conclude that for any fixed $L$, there is $M$ large enough so that $B_{L}^c \cap C_{M} = \emptyset$.

Choose $L$ such that $\bP(B_{L}) < \epsilon/2$. That fact that we can do so is guaranteed by \eqref{equate1}. Choose $M$ such that $B_{L}^{c} \cap C_M = \emptyset$ on this set.

We then have, on $B_{L}^c = B_{L}^{c} \cap C_{M}^{c}$,
\begin{align} \label{ChainOfInequalities}
\sqrt{n} \hat{h}^2 & \leq 16 \cdot L \cdot \hat{h}^{1/2} (1 + I_{0} + \hat{I})^{1/2} + 4 \lambda_n \sqrt{n} (I_{0} - \hat I) \nonumber \\
& \leq 16 \cdot L \cdot \hat{h}^{1/2}(1 + I_{0} + M)^{1/2} + 4 \lambda_n \sqrt{n} I_0 \\ 
& \leq 2 \max \left\{ 16 \cdot L \cdot \hat{h}^{1/2} (1+ I_{0} + M)^{1/2}, 4 \lambda_n \sqrt{n} I_0 \right\}. \nonumber
\end{align}
From this, we conclude
$$\hat{h} \leq \max\{n^{-1/3} \cdot (32 L)^{2/3} \cdot (1+ I_{0} + M)^{1/3}, \sqrt{\lambda_n \cdot 8 \cdot I_0} \}$$
on $B_{L}^c$. Choose $K$ so that $\bP(\lambda n^{2/3} > K) < \epsilon/2$, which is permitted because $\lambda = O_{P}(n^{-2/3})$. We have
\begin{equation} \label{ConclusionCaseTwo}
\hat{h} \cdot n^{1/3} \leq \max\{(32 L)^{2/3} (1 + I_{0} + M)^{1/3}, \sqrt{8 \cdot K \cdot I_{0}}\}.
\end{equation}

The right hand side is constant, depending on the choice of $\epsilon$. The set on which this bound does not hold has probability less than $\epsilon$, by the choice of $B_{L}$ and $K$.

Having examined probabilistic rates for both cases $\hat{h} \leq n^{-1/3}(1+\hat I+I_{0})$ and $\hat{h} > n^{-1/3}(1+\hat I+I_{0})$, we turn next to proving the same rate for squared Hellinger risk. This requires a more refined application of Theorems \ref{InCor} and \ref{OutThm}. 


We will show that there exist $n_{0} \in \mathbb{N}$ and $c \geq 0$ such that $n \geq n_{0}$ implies $\mathbb{E}_{f_{0}}[\hat{h}^2 n^{2/3}] \leq c$. We have 
\begin{equation}\label{BeginningEquation}
\mathbb{E}_{f_{0}}[\hat{h}^2 n^{2/3}] = \mathbb{E}_{f_0}[\hat{h}^2 n^{2/3} (\mathbbm{1}_{\hat{h} \leq n^{-1/3}(1+\hat{I}+I_{0})} + \mathbbm{1}_{\hat{h} > n^{-1/3}(1+\hat{I}+I_{0})})].
\end{equation}
We will consider both terms in this summand individually. First, we have
$$\mathbb{E}_{f_{0}}[\hat{h}^2 n^{2/3} \mathbbm{1}_{\hat{h} \leq n^{-1/3}(1+\hat{I}+I_{0})}] = \int_{0}^{\infty} \bP(\hat{h}^2 n^{2/3} \mathbbm{1}_{\hat{h} \leq n^{-1/3}(1+\hat{I}+I_{0})} \geq u) \, du.$$
From \eqref{RandomEquation},
\begin{equation} \label{Equation1}
\int_{0}^{\infty} \bP(\hat{h}^2 n^{2/3} \mathbbm{1}_{\hat{h} \leq n^{-1/3}(1+\hat{I}+I_{0})} \geq u) \, du \leq \int_{0}^{\infty} \bP(4 \lambda_{n} (1+2I_{0}) n^{2/3} \geq u) \, du.
\end{equation}
From the definition of $\lambda_n$ and Theorem \ref{InCor}, 
$$\bP(4 \lambda_n(1+2I_{0}) n^{2/3} \geq u) \leq c_{0} \exp\left[-\frac{u}{4(1+2I_{0}) c_{0}^2}\right]$$
for $n$ and $u$ large. This allows us to integrate the right-hand side of \eqref{Equation1}, which gives that $\mathbb{E}_{f_0}\left[\hat{h}^2 n^{2/3} \mathbbm{1}_{\hat{h} \leq n^{-1/3}(1+\hat{I}+I_{0})} \right]$ is finite.

On the other hand, consider the second expectation $\mathbbm{E}\left[\hat{h}^2 n^{2/3} \mathbbm{1}_{\hat{h} > n^{-1/3}(1+\hat{I}+I_{0})}\right]$. Again we have
\begin{equation} \label{Equation2}
\mathbb{E}_{f_0} \left[\hat{h}^2 n^{2/3} \mathbbm{1}_{\hat{h} > n^{-1/3}(1+\hat{I}+I_{0})}\right] = \int_{0}^{\infty} \bP(\hat{h}^2 n^{2/3} \mathbbm{1}_{\hat{h} > n^{-1/3}(1+\hat{I}+I_{0})} \geq u) \, du.
\end{equation}
Denote by $A_{u}$ the event that $\{\hat{h}^2 n^{2/3} \mathbbm{1}_{\hat{h} > n^{-1/3}(1+\hat{I}+I_{0})} \geq u\}$. Let $B_L$ and $C_M$ be as in \eqref{B_L}-\eqref{C_M}, and denote by $\Lambda_K$ the event $\{\lambda_n n^{2/3} > K\}$. Choosing $L = \left(\frac{u^3}{3\cdot 32^2}\right)^{1/7}$, $M = L^5$, and $K = \frac{u^2}{8I_0}$ gives \eqref{ChainOfInequalities} for large enough $u$. Furthermore, both of the arguments in the maximum of \eqref{ConclusionCaseTwo} are less than $u$. Recalling that $B_{L}^{c} \cap C_M = B_{L}^{c}$, this gives
\begin{align*}
\bP(A_{u}) & \leq \bP(A_{u} \cap B_{L}) + \bP(A_{u} \cap B_{L}^{c} \cap \Lambda_{K}^{c}) + \bP(A_{u}\cap B_{L}^{c} \cap \Lambda_{K}) \\
& \leq \bP(B_{L}) + 0 + \bP(\Lambda_K) \\
& \leq c \exp \left[-\frac{L}{c^2}\right]+c_{0} \exp\left[-\frac{K}{c_{0}^2}\right].
\end{align*}
This last inequality follows from Theorems \ref{InCor} and \ref{OutThm}. The fact that $\bP(A_{u} \cap B_{L}^{c} \cap \Lambda_K)$ equals zero follows from \eqref{ConclusionCaseTwo} and our choice of $L$, $M$, and $K$. Therefore the expectation in \eqref{Equation2} is finite. Since we have shown that both of the expectations in \eqref{BeginningEquation} are bounded by constants for $n_0$ large enough, the result is proven.
\qed

\emph{Proof of the Basic Inequality, Lemma \ref{SketchBasic}}.

We have
\begin{align*}
4 P_{n}(p_{\hat{f}_{n}}) - \lambda_n \hat{I} &= 2 \int \log\left( \frac{\hat{f}_{n} + f_{0}}{2 f_{0}} \right) \, d P_{n} - \lambda_n \hat{I} \\
&\geq \int \log\left( \frac{\hat{f}_{n}}{f_{0}} \right) dP_{n} - \lambda_n \hat{I} \\
&\geq - \lambda I_{0}
\end{align*}
The first inequality comes from the concavity of $\log$. The second is from the definition of $\hat{f}_{n}$ as the minimizer of $-\int \log f \, d P_n + \lambda_{n} I(f)$, which implies $-\int \log \hat{f}_{n} \, d P_n + \lambda_{n} \hat{I} \leq -\int \log f_{0} \, d P_n + \lambda_{n} I_{0}$. 

We also have that 
\begin{align*}
-16 \int p_{\hat{f}_{n}} d P & \geq 16 h^2\left( \frac{\hat{f}_{n} + f_{0}}{2}, f_{0} \right) \\
& \geq h^2(\hat{f}_{n}, f_{0}),
\end{align*}
by Lemmas 4.1 and 4.2 in \cite{vandeGeer}.

Therefore,
\begin{align*}
16 \int p_{\hat{f}_{n}} \, d(P_{n}-P) - 4 \lambda_{n} \hat{I} &\geq -16 \int p_{\hat{f}_{n}} \, dP - 4 \lambda_{n} I_{0} \\
&\geq 16 h^2\left(\frac{\hat{f}_{n} +f_{0}}{2}, f_{0}\right) - 4 \lambda_{n} I_{0} \\
& \geq h^{2}(\hat{f}_{n}, f_{0}) - 4 \lambda_{n} I_{0}
\end{align*}
This proves the result.
\qed

\emph{Proof of Theorem \ref{NetRate}}.

Recall that the total-variation on a geometric network $G$ is the sum of the total variation over the edges. In this proof only, we denote the graph-induced total variation by $\TV_{G}$ and univariate total variation $\TV$, because both will be used in similar contexts. Let $I_{G}(f) = \TV_{G}(\log f)$, $\hat{I}_{G} = I_{G}(\hat{f}_{n})$, and $I_{0,G} = \TV_{G}(\log f_{0})$.

From the univariate proof, we have that
\begin{equation} \label{seven}
\sup_{h(f, f_{0}) > n^{-1/3}(1+I(f)+I{0})} \frac{n^{1/2} \abs{(P_{n} -P)(p_f)}}{h^{1/2}(f, f_{0})(1+I(f)+I_{0})^{1/2}} = O_{P}(1) 
\end{equation}
and
\begin{equation}
\sup_{h(f, f_{0}) \leq n^{-1/3}(1+I(f)+I_{0})} \frac{n^{2/3} \abs{(P_{n} -P)(p_{f})}}{1+I(f)+I_{0}} = O_{P}(1).
\end{equation}
We proceed analogously to Theorem \ref{RateResult}. Take
$$\lambda_{n} = \max\left\{\sup_{h(f,f_{0}) \leq n^{-1/3}(1+I(f)+I_{0})} \frac{8 (P_{n}-P)(p_{f})}{1+I(f)+I_{0}}, n^{-2/3} \right\}.$$

From the Basic Inequality, Lemma \ref{SketchBasic}, we have
\begin{align}
\hat{h}^2 \leq \max \left\{\mathbbm{1}_{\hat{h} > n^{-1/3}(1+\hat{I}+I_{0})} \left( 16 (P_{n} - P)(p_{\hat{f}_{n}}) + 4 \lambda_{n}(I_{0,G} - \hat{I}_{G}) \right) \right. \label{1}\\
\left. \mathbbm{1}_{\hat{h} \leq n^{-1/3}(1+\hat{I}+I_{0})} \left( 16 (P_{n} - P)(p_{\hat{f}_{n}}) + 4 \lambda_{n}(I_{0,G} - \hat{I}_{G}) \right) \right\} \label{2}
\end{align}

First, consider the case $\hat{h} > n^{-1/3}(1+\hat{I}+I_{0})$. Define subsets of the probability space
$$B_{L} = \left\{\sqrt{n} \abs{ (P_{n}-P)(p_{\hat f})} > L \cdot \hat{h}^{1/2} \cdot(1+\hat I+I_{0})^{1/2} \right\}$$
and
$$C_{M} = \{\hat{I} > M \geq 2 I_{0,G}\}.$$

Because $\hat{I} \leq 2 \hat{I}_{G}$ (Lemma \ref{TVBounds}), on $C_{M}$ we have $I_{0,G} - \hat I_{G} < 0$ on $C_M$. Proceeding as in the proof of Theorem \ref{RateResult}, the fact that $\lambda_{n} n^{-2/3}$ is bounded below by $1$ gives that on $B_{L}^{c}$
$$\sqrt{n}h^{3/2} \leq 16 \cdot L \cdot (1+\hat{I}+I_{0})^{1/2} + \frac{4 (I_{0, G}-\hat{I}_{G})}{\left(16 \cdot L \cdot (1+\hat{I}+I_{0})^{1/2}\right)^{1/3}}.$$

As $M$ gets large, this inequality and the fact that $2 \hat{I}_{G}$ dominates $\hat{I}$ gives that $B_{L}^{c} \cap C_{M} = \emptyset$. So on $B_L^{c}$, for large enough $M$,
\begin{align*}
\sqrt{n} \hat{h}^2 &\leq 16 \cdot L \cdot \hat{h}^{1/2}(1+I_{0} + \hat{I})^{1/2} + 4 \lambda_{n} \sqrt{n}(I_{0,G}-\hat I_{G})\\
& \leq 2 \max \left\{ 16 \cdot L \cdot \hat{h}^{1/2}(1+I_{0}+M)^{1/2}, 2 \lambda_{n} \sqrt{n} I_{0,G} \right\}.
\end{align*}
This holds with probability $1-\epsilon$ if we choose $L$ so that $B_L$ holds with probability less than $\epsilon$--the fact that we can do so is guaranteed by \eqref{seven}. We conclude that when $\hat{h} > n^{-1/3}(1+\hat{I}+I_{0})$, $\hat{h}^2 = O_{P}\left(\max\{\lambda_n, n^{-2/3}\}\right)$.

Next consider the case $\hat{h} \leq n^{-1/3}(1+\hat{I}+I_{0})$. Mirroring equations \eqref{6}-\eqref{3}, we have
\begin{align*}
\hat{h}^2 &\leq 4(1+I_{0,G}+\hat{I}_{G})\left(\sup_{h(f,f_{0}) \leq n^{-1/3}(1+I+I_{0})} \frac{4 (P_{n}-P)(p_{f})}{1+I_{0,G}+I_{G}}-\lambda_{n}\right)+4\lambda_{n}(1+2I_{0,G}) \\
& \leq 4(1+I_{0,G}+\hat{I}_{G})\left(\sup_{h(f,f_{0}) \leq n^{-1/3}(1+I+I_{0})} \frac{8 (P_{n}-P)(p_{f})}{1+I_{0}+I}-\lambda_{n}\right)+4\lambda_{n}(1+2I_{0,G})
\end{align*}
The last inequality again comes from $1+I_{0}+I \leq 2(1+I_{0,G}+I_{G})$. By our choice of $\lambda_n$ we have that
$$\hat{h}^2 \leq 4 \lambda_{n}(1+2I_{0,G}).$$
Because $\lambda_n = O_{P}(n^{-2/3})$, we have that when $\hat{h} \leq n^{-1/3}(1+\hat{I}+I_{0})$, $\hat{h}^2 = O_{P}\left(n^{-2/3}\right)$. Having established the probabilistic rate for both cases of $\hat{h}$, we must now translate these into rates for the squared Hellinger risk. In the unvariate case, we use the probabilistic bounds just derived--for the cases $\hat{h} \leq n^{-1/3}(1 + \hat{I} + I_0)$ and $\hat{h} > n^{-1/3}(1+\hat{I}+I_0)$--to prove an equivalent rate in Hellinger risk. This part of the proof follows exactly as in the analogous result for the univariate case, and as such is omitted.
\qed

\subsection{Empirical Process Results} \label{subsec2}
The goal of this section is to prove the following statements, Theorems \ref{ThmSktch1} and \ref{ThmSktch2}, which were used in the proof of Theorem \ref{RateResult}.
\begin{equation}
\sup_{h(f, f_{0}) \leq n^{-1/3}(1+I(f)+I_{0})} \frac{n^{2/3} \abs{(P_{n}-P)(p_{f})}}{1+I(f)+I_{0}} = O_{P}(1)
\end{equation}
\begin{equation}
\sup_{h(f, f_{0}) > n^{-1/3}(1+I(f)+I_{0})} \frac{n^{1/2} \abs{(P_{n}-P)(p_{f})}}{h^{1/2}(f, f_{0})(1+I(f)+I_{0})^{1/2}} = O_{P}(1) 
\end{equation}
These are the simplifications of the results of Theorems \ref{InCor} and \ref{OutThm}, respectively. We begin by introducing notation and relevant definitions.

The \emph{Bernstein Difference} for a parameter $K \in \N$, is given by $\rho_{K}$, where
$$\rho_{K}^{2}(g) = 2 K^2 \int \left( e^{\abs{g}/K} - 1 - \abs{g}/K \right) \, d P$$

\emph{Generalized entropy with bracketing}, denoted $\mathcal{H}_{B,K}$ is entropy with bracketing,      where the $L_{2}(P)$ metric is replaced by the Bernstein difference $\rho_{K}$. $H_{B}$ denotes the usual entropy with bracketing.

The following theorem is an important tool at our disposal.
\begin{theorem}[\cite{vandeGeer}, 5.11] \label{5.11}
Let $\mathcal{G}$ be a function class which satisfies
$$\sup_{g \in \mathcal{G}} \rho_{K}(g) \leq R.$$
Then there is a universal constant $C$ such that for any $a$, $C_{0}$, $C_{1}$ which satisfy
\begin{align}
a &\leq C_{1} \sqrt{n} R^2 /K, \\
a &\geq C_0 \left(\max\left\{\int_{0}^{R} \mathcal{H}_{B,K}^{1/2}(u, \mathcal{G}, P) \, du, \; R \right\} \right), \label{BrackInt} \\
C_{0}^2 &\geq C^2(C_{1} + 1),
\end{align}
we have
$$\bP\left(\sup_{g \in \mathcal{G}} \abs{\sqrt{n} (P_{n} - P)(g)} \geq a\right) \leq C          \exp\left[-\frac{a^2}{C^2(C_{1} + 1) R^2} \right].$$
\end{theorem}
Our statement of Theorem \ref{5.11} is a simplification of the full statement in the listed reference. Because we only work with bracketing entropy integrals which are convergent, we simplify according to the author's comments following the theorem, omitting the second condition in the full statement and taking the lower bound in the bracketing entropy integral in \eqref{BrackInt} to be zero.

The following lemmas will also be required.
 
\begin{lemma}[\cite{vandeGeer}, 5.8]
Suppose that
$$\norm{g}_{\infty} \leq K$$
and
$$\norm{g}_{2} \leq R.$$
Then
$$\rho_{2K}(g) \leq \sqrt{2} R.$$
\end{lemma}

\begin{lemma}[\cite{vandeGeer}, 5.10] \label{5.10}
Suppose $\mathcal{G}$ is a set of functions such that
$$\sup_{g \in \mathcal{G}} \norm{g}_{\infty} \leq K.$$
Then
$$\mathcal{H}_{B,4K}(\sqrt{2}\delta, \mathcal{G}, P) \leq H_{B}(\delta, \mathcal{G}, P) \; \text{ for all } \delta > 0.$$
\end{lemma}

\begin{lemma}[\cite{vandeGeer}, 7.2 \& 4.2] \label{7.2}
Let $p_{f}$ be of the form $p_{f} = \frac{1}{2} \log \frac{f+f_{0}}{2 f_{0}}$, as occurred in Lemma \ref{SketchBasic}. Then
$$\rho_{1}(p_{f}) \leq 4 h\left(\frac{f+f_{0}}{2}, f_{0}\right) \leq \frac{ 4 h(f, f_{0})}{\sqrt{2}}$$
\end{lemma}

\begin{lemma} \label{rhorelate}
Let $L$ and $K$ be natural numbers such that $L > K$. Then for any function $g$, $\rho_{K}(g) \geq \rho_{L}(g)$.
\end{lemma}
\begin{proof}
From the Taylor series expansion of $e^x$,
\begin{align*}
\rho^{2}_{K}(g) &= 2 K^2 \int \left( e^{\abs{g}/K}-1-\abs{g}/K \right) \, dP \\
&= 2 \int K^2 \sum_{m=2}^{\infty} \frac{\abs{g}^m}{m! \cdot K^m} \, dP \\
&= 2 \int \sum_{m=2}^{\infty} \frac{\abs{g}^m}{m! \cdot K^{m-2}} \, dP \\
&\geq 2 \int \sum_{m=2}^{\infty} \frac{\abs{g}^m}{m! \cdot L^{m-2}} \, dP \\
&= \rho^{2}_{L}(g)
\end{align*}
\end{proof}

This last lemma is a culmination of new results on bracketing entropy. Its proof can be found in Appendix \ref{AppB}, along with other contributions on bracketing entropy of function classes with uniformly bounded variation. We denote the quantity $1+I(f)+I_{0}$ by $J(f)$.

\begin{lemma} \label{CoverAndBound}
The set of functions $\mathcal{P}_{M} = \{p_{f}: J(f) \leq M\}$ satisfies, for some constant $A$,
$$H_{B}(\delta, \mathcal{P}_{M}, P) \leq A \cdot \frac{M}{\delta}, \; \; \forall \delta > 0.$$
Furthermore, $p_{f} \in \mathcal{P}_{M}$ implies $\norm{p_{f}}_{\infty} < M$.
\end{lemma}

With these lemmas in hand, we are ready to state and prove our main results. We will prove a sequence of constrained results, and then use a peeling device to obtain the concentration inequalities. The method of proof, and particularly our use of the peeling device, is interesting in its own right. Our first result is Lemma \ref{LemSktch3}, which establishes bounds for the supremum of the empirical process indexed by $\{f: J(f) \leq M \text{ and } h(f,f_{0}) \leq \delta\}$ for constants $\delta$ and $M$. 

\emph{Proof of Lemma \ref{LemSktch3}.}

By Lemma \ref{7.2}, $h(f,f_{0}) \leq \delta$ gives that $\rho_{1}(p_{f}) \leq \frac{4}{\sqrt{2}} \delta = 2^{3/2} \delta$. 

By Lemma \ref{rhorelate}, $\rho_{1}(p_{f}) \leq 2^{3/2} \delta$ gives that $\rho_{4M}(p_{f}) \leq 2^{3/2} \delta$ for all $M \geq 1$.

By Lemma \ref{CoverAndBound}, $p_{f} \in \mathcal{P}_{M}$ gives that $\norm{p_{f}}_{\infty} \leq M$. From Lemmas \ref{5.10} and \ref{CoverAndBound}.
$$\mathcal{H}_{B,4M}(\delta, \mathcal{P}_{M}, P) \leq H_{B}(\delta/\sqrt{2}, \mathcal{P}_{M}, P) \leq \frac{A \sqrt{2} M}{\delta}.$$

Collecting these facts, we seek to apply Theorem \ref{5.11}. We have $\rho_{4M}(p_{f}) \leq 2^{3/2} \delta$. From the conditions in the theorem (with $R = 2^{3/2} \delta$, $K = 4M$ and $a = 2^{-1/2} C_{1} \sqrt{M} \delta^{1/2}$), it suffices to choose $\delta, C_{0}$, $C_{1}$ such that
\begin{align}
a &\leq C_{1} \sqrt{n} \frac{2^3 \delta^2}{4 M} = \frac{2C_{1} \sqrt{n} \delta^2}{M} \label{1st}\\
a &\geq C_{0} \int_{0}^{R} H_{B}^{1/2} (u/\sqrt{2}, \mathcal{P}_{M}, P) = 2 C_0 \sqrt{A M \delta} \label{3rd}\\ 
C_{0}^{2} &\geq C^{2}(C_{1}+1) \label{4th}
\end{align}

Choose $C_{1} = 2 C_0 \sqrt{2 A}$. Then \eqref{1st} is satisfied for $\delta \geq \frac{M}{2} \cdot n^{-1/3}$, $\eqref{3rd}$ is satisfied by the choice of $a$, and $\eqref{4th}$ is satisfied for large enough $C_{0}$. By Theorem \ref{5.11} we have for all $\delta \geq \frac{M}{2} \cdot n^{-1/3}$ (if $C_{1} \geq 1$)
\begin{align*}
\bP\left(\sup_{p_{f} \in \mathcal{P}_{M}, h(f,f_{0}) \leq \delta} \abs{\sqrt{n} (P_{n}-P)(p_{f})} \geq 2 C_{1} \sqrt{M} \delta^{1/2} \right) &\leq C \exp\left[-\frac{4 C_{1}^2 M \delta}{C^2(C_1+1) 2^3 \delta^2} \right]\\
&\leq C \exp \left[-\frac{C_{1} M \delta^{-1}}{4 C^2} \right]
\end{align*}
\qed

\begin{theorem} \label{OutThm}
There are constants $c$, $n_{0}$ and $t_0$ so that when $n \geq n_{0}$ and $T \geq t_0$
$$\bP\left(\sup_{p_f \in P, h(f,f_{0}) > n^{-1/3} J(f)} \frac{\abs{\sqrt{n} (P_{n} - P)(p_{f})}}{h^{1/2}(f,f_{0}) J^{1/2}(f)}\geq T \right) \leq c \exp \left[-\frac{T}{c^2}\right].$$
\end{theorem}

\begin{proof}
We first prove the following: there are constants $n_{0}$, $t_{0}$ and $c_{0}$ such that for all $n \geq n_{0}$, $T \geq t_{0}$, and $M \geq 1$
\begin{equation}\label{WasALemma}
\bP\left(\sup_{p_f \in \mathcal{P}_{M}, h(f,f_{0}) > \frac{M}{2} n^{-1/3}} \frac{\abs{\sqrt{n} (P_{n} - P)(p_{f})}}{h^{1/2}(f,f_{0})} \geq T \sqrt{\frac{M}{2}} \right) \leq c_{1} \exp\left[- \frac{T M}{c^{2}_{1}}\right].
\end{equation}

The proof of this claim is an application of the peeling device \cite{vandeGeer}[Section 5.3] to Lemma \ref{LemSktch3}. Let $S = \min\{s \in \N: 2^{-s} < \frac{M}{2} n^{-1/3}\}$. We will form a union bound by partitioning into sets with $\{2^{-s-1} < h(f,f_0) \leq 2^{-s}\}$ for integer-valued $s$. Because Hellinger distance is bounded above by $1$, we need not consider negative values of $s$. Let $T = 4 C_{1}$. Applying this union bound, we have
\begin{align*}
& \bP\left(\sup_{p_f \in \mathcal{P}_{M}, h(f,f_{0}) > \frac{M}{2} n^{-1/3}} \frac{\abs{\sqrt{n} (P_{n} - P)(p_{f})}}{h^{1/2}(f,f_{0})} \geq T \sqrt{\frac{M}{2}} \right) \\
\leq  & \sum_{s=1}^{S} \bP \left(\sup_{p_f \in \mathcal{P}_{M}, 2^{-s} < h(f,f_{0}) \leq 2^{-s+1}} \frac{\abs{\sqrt{n} (P_{n} - P)(p_{f})}}{h^{1/2}(f, f_{0})} \geq T \sqrt{\frac{M}{2}} \right) \\
\leq & \sum_{s=1}^{S} \bP \left(\sup_{p_f \in \mathcal{P}_{M}, h(f,f_{0}) \leq 2^{-s+1}}  \abs{\sqrt{n} (P_{n} - P)(p_{f})} \geq 2^{\frac{-s+1}{2}} \cdot 2 C_{1} \sqrt{M} \right) 
\end{align*}
We have $2^{-s+1} \geq \frac{M}{2} n^{-1/3}$ for $s \leq S$, so applying Lemma \ref{LemSktch3} gives the further bound
\begin{align}
& \leq  \sum_{s=1}^{S} C \cdot \exp \left[-\frac{C_{1} \cdot M \cdot (2^{-s+1})^{-1}}{4 C^2}\right] \nonumber \\
& =  \sum_{s=1}^{S} C \cdot \exp \left[-\frac{C_{1} \cdot M \cdot 2^{s-1}}{4 C^2}\right] \nonumber\\
& \leq  \sum_{s=1}^{S} C \cdot \exp\left[-\frac{C_{1} M}{8 C^2} - \frac{2^{s-2}}{4C^2}\right] \label{algebra}\\
& \leq  \exp\left[-\frac{C_{1} M}{8 C^2} \right] \sum_{s=1}^{S} C \exp\left[-\frac{2^{s-2}}{4C^2}\right] \nonumber\\ 
& =  c_{1} \exp\left[-\frac{T M}{c_{1}^{2}}\right]. \nonumber
\end{align}
Here, $c_1$ is some constant, since the final summation is convergent as $S$ approaches infinity. The third inequality in this chain follows from $C_{1} M 2^{s-1} \geq \frac{M C_{1}}{2} + M C_{1} 2^{s-2}$, so that when $M \geq 1$ and $C_{1} \geq 1$, 
$$C_1 M 2^{s-1} \geq \frac{M C_{1}}{2} + 2^{s-2}.$$
Of course, it suffices to consider $M \geq 1$ because $J(f) \geq 1$. This proves the claim.

We use the claim to prove the result by again applying the peeling device, but this time with respect to $J(f)$. Because $J(f) \geq 1$, we need only peel in sets $\{2^s \leq J(f) \leq 2^{s+1}\}$ for $s \geq 0$. This gives
\begin{align*}
& \bP\left( \sup_{p_f \in \mathcal{P}, h(f,f_{0}) > n^{-1/3} J(f)} \frac{\abs{\sqrt{n} (P_{n} - P)(p_{f})}}{h^{1/2}(f,f_0) J^{1/2}(f)} \geq T\right) \\
\leq &\sum_{s=0}^{\infty} \bP\left( \sup_{p_f \in \mathcal{P}, J(f) \leq 2^{s+1}, h(f,f_0) > n^{-1/3} 2^{s}} \frac{\abs{\sqrt{n} (P_{n} - P)(p_{f})}}{h^{1/2}(f,f_0)} \geq T 2^{s/2}\right).
\end{align*}
Applying the claim and manipulating as in \eqref{algebra}, there is a constant $c$ which permits the following bound.
\begin{align*}
\leq &\sum_{s=0}^{\infty} c_1 \exp \left[-\frac{T 2^{s+1}}{c_{1}^{2}}\right] \\
\leq &\exp \left[-\frac{T}{2 c_{1}^{2}}\right] \sum_{s=0}^{\infty} c_1 \exp\left[-\frac{2^s}{c_{1}^2} \right] \\
\leq &c \exp\left[-\frac{T}{c^2}\right].
\end{align*}
\end{proof}

\begin{theorem} \label{InCor}
There are constants $n_{0}$, $t_{0}$, and $c$ such that for all $n \geq n_{0}$ and $T \geq t_0$
$$\bP\left(\sup_{p_{f} \in \mathcal{P}, h(f,f_0) \leq n^{-1/3} J(f)} \frac{\abs{n^{2/3} (P_n - P)(p_{f})}}{J(f)} \geq T \right) \leq c_{0} \exp \left[- \frac{T}{c_{0}^{2}}\right]$$
\end{theorem}

\begin{proof}
First we apply the peeling device to the quantity $J(f)$. We partition into sets with $2^s < J(f) \leq 2^{s+1}$. Since $J(f) \geq 1$, it suffices to take $s \geq 0$. We have
\begin{align*}
& \bP\left(\sup_{p_{f} \in \mathcal{P}, h(f,f_0) \leq n^{-1/3} J(f)} \frac{\abs{\sqrt{n}(P_{n} - P)(p_{f})}}{J(f) n^{-1/6}} \geq T\right) \\
=& \bP\left(\sup_{p_{f} \in \mathcal{P}, h(f,f_0) \leq n^{-1/3} J(f)} \frac{\abs{\sqrt{n}(P_{n} - P)(p_{f})}}{\sqrt{J(f)} \sqrt{J(f) n^{-1/3}}} \geq T\right) \\
\leq &\sum_{s=0}^{S} \bP\left(\sup_{p_f \in \mathcal{P}, \, h(f,f_0) \leq n^{-1/3} J(f), \, 2^s \leq J(f) \leq 2^{s+1}} \frac{\abs{\sqrt{n} (P_n - P)(p_f)}}{\sqrt{J(f)} \sqrt{J(f) n^{-1/3}}} \geq T \right) 
\end{align*}
We peel this expression in $h(f,f_0)$. For $s \in \mathbb{N}$, let $R_{s} = \max\{r \in \mathbb{N}: 2^{-r} \geq n^{-1/3} 2^{s+1}\}$. Let
$$\mathcal{N}_{s,r} = \{p_f \in \mathcal{P}, \, 2^{-r-1} < h(f,   f_0) \leq 2^{-r}, \, 2^s \leq J(f) \leq 2^{s+1}, \, h(f,f_0) \leq n^{-1/3} J(f)\}$$
 for $r = 0,...,R_s-1$ and
$$\mathcal{N}_{s,R_{s}} = \{p_f \in \mathcal{P}, \, h(f,f_0) \leq 2^{s+1} n^{-1/3} \leq 2^{-R_s}, \, 2^s \leq J(f)   \leq 2^{s+1}, \, h(f,f_0) \leq n^{-1/3} J(f)\}.$$
Applying the peeling device gives the further bound.
\begin{align} \label{summand}
\leq &\sum_{s=0}^{\infty} \left\{ \sum_{r=0}^{R_s-1} \bP\left(\sup_{\mathcal{N}_{s,r}} \frac{\abs{\sqrt{n} (P_n - P)(p_f)}}{\sqrt{J(f)} \sqrt{J(f) n^{-1/3}}} \geq T \right) + \bP\left(\sup_{\mathcal{N}_{s,R_{s}}} \frac{\abs{\sqrt{n} (P_n - P)(p_f)}}{\sqrt{J(f)} \sqrt{J(f) n^{-1/3}}} \geq T \right) \right\}.
\end{align}
In this last term, $J(f) n^{-1/3}$ and $J(f)$ can be bounded below on $\mathcal{N}_{s,R_{s}}$. Indeed, $J(f) > 2^{s}$ and 
$$J(f) n^{-1/3} > 2^s n^{-1/3} > 2^{-R_{s}-2}.$$
Inserting these bounds gives
\begin{align*}
&\bP\left(\sup_{\mathcal{N}_{s,R_s}} \frac{\abs{\sqrt{n} (P_n - P)(p_f)}}{\sqrt{J(f)} \sqrt{J(f) n^{-1/3}}} \geq T \right)\\
\leq &\bP\left(\sup_{p_f \in \mathcal{P}_{2^{s+1}}, \, h(f,f_0) \leq 2^{-R_s}} \frac{\abs{\sqrt{n} (P_n - P)(p_f)}}{2^{s/2} 2^{-(R_s+2)/2}} \geq T \right).
\end{align*}
Applying Lemma \ref{LemSktch3} (with $T=4\sqrt{2} C_1$) bounds this term by an expression of the form $c_{1} \exp\left[-\frac{T 2^{s+1} 2^{R_{s}}}{c_{1}^2}\right]$, for some constant $c_1$.
For $r < R_s$, we have the following chain of inequalities
\begin{align*}
& \bP\left(\sup_{N_{s,r}} \frac{\abs{\sqrt{n} (P_n - P)(p_f)}}{\sqrt{J(f)} \sqrt{J(f) n^{-1/3}}} \geq T \right)\\
\leq & \bP\left(\sup_{p_f \in \mathcal{P}_{2^{s+1}}, \, 2^{-r-1} < h(f,f_0) \leq 2^{-r}} \frac{\abs{\sqrt{n} (P_n - P)(p_f)}}{2^{s/2} h^{1/2}(f,f_0)} \geq T \right)\\
\leq & \bP\left(\sup_{p_f \in \mathcal{P}_{2^{s+1}}, \, h(f,f_0) \leq 2^{-r}} \frac{\abs{\sqrt{n} (P_n - P)(p_f)}}{2^{s/2} 2^{-(r+1)/2}} \geq T \right).
\end{align*}
According to the definition of $R_s$, $2^{-r} \geq 2^{s+1} n^{-1/3}$, so Lemma \ref{LemSktch3} (with $T = 4 C_1$) allows us to bound this probability by $c_2 \exp\left[-\frac{T 2^{s+1} 2^{r}}{c_{2}^2}\right]$.

The double summand \eqref{summand} is thus bounded by
$$\sum_{s=0}^{\infty} \left\{\sum_{r=0}^{R_s-1} c_{1} \exp \left[-\frac{T 2^{s+1} 2^{r}}{c_{1}^2}\right] + c_2 \exp\left[-\frac{T 2^{s+1} 2^{R_s}}{c_{2}^2}\right] \right\}.$$

Reducing twice according to the manipulation in \eqref{algebra}, this expression can be bounded by a term of the form $c_0 \exp\left[-\frac{T}{c_{0}^2}\right]$.
\end{proof}

\subsection{Depth-First Embedding a Geometric Network into $\R$}

The goal of this section is to define an embedding $\gamma$, of a fixed geometric network $G$ into $\R$, which approximately preserves total variation.
Let $g$ be a function of bounded variation on $G$. On each edge $e$, which is identified with the interval $[0, l_{e}]$, 
\begin{equation} \label{TVeq}
\TV(g_{e}) = \abs{g_{e}(0) - \lim_{x \searrow 0} g_{e}(x)} + \TV(g_{e}|_{(0, l_{e})}) + \abs{g_{e}(l_{e}) - \lim_{x \nearrow l_{e}} g(x)}.
\end{equation}
Define 
$$\tilde{g}_{e}(x) = \begin{cases} \lim_{z \searrow 0} g_{e}(z) & \text{if } x = 0 \\
g_{e}(x) & \text{if } x \in (0, l_{e}) \\
\lim_{z \nearrow l_{e}} g_{e}(z) & \text{if } x = l_{e}
\end{cases}.$$
The fact that $g$ is of bounded variation gives that these limits exist. 
Furthermore,
$$\TV(g_{e}|_{(0,l_{e})}) = \TV(\tilde{g}_{e}|_{[0,l_{e}]}).$$
We can therefore rewrite \eqref{TVeq} as
$$\TV(g_{e}) = \abs{g_{e}(0) - \tilde{g}_{e}(0)} + \TV(\tilde{g}_e) + \abs{g_{e}(l_{e}) - \tilde{g}_{e}(l_{e})}.$$
From this equation we conclude that, repeating this procedure for all edges $e$
\begin{equation} \label{TVeq2}
\TV_{G}(g) = \sum_{e \in E} \abs{g_{e}(0) - \tilde{g}_{e}(0)} + \TV(\tilde{g}_{e}) + \abs{g_{e}(l_{e}) - \tilde{g}_{e}(l_{e})}.
\end{equation}
Equation \eqref{TVeq2} gives us the following insight: the total variation of a function on a network can be decomposed into jumps at nodes and total variation along open intervals. By separating limit nodes from the true value at the node, we can create an expanded network that represents this decomposition. For each node, and each edge incident to that node, define a limit node as the limit approaching the node along the incident edge. Similarly, define a value node as the true value at a node. The expanded graph is the defined on the limit and value nodes, with the inherited connectivity. Open intervals corresponding to the original edges in the geometric graph are edges between two limit nodes, and value nodes are only connected to limit nodes. We perform this expansion in order to guarantee that each edge in the original network is traversed by a depth-first search.

\begin{figure}[ht!]
\includegraphics[width=12cm]{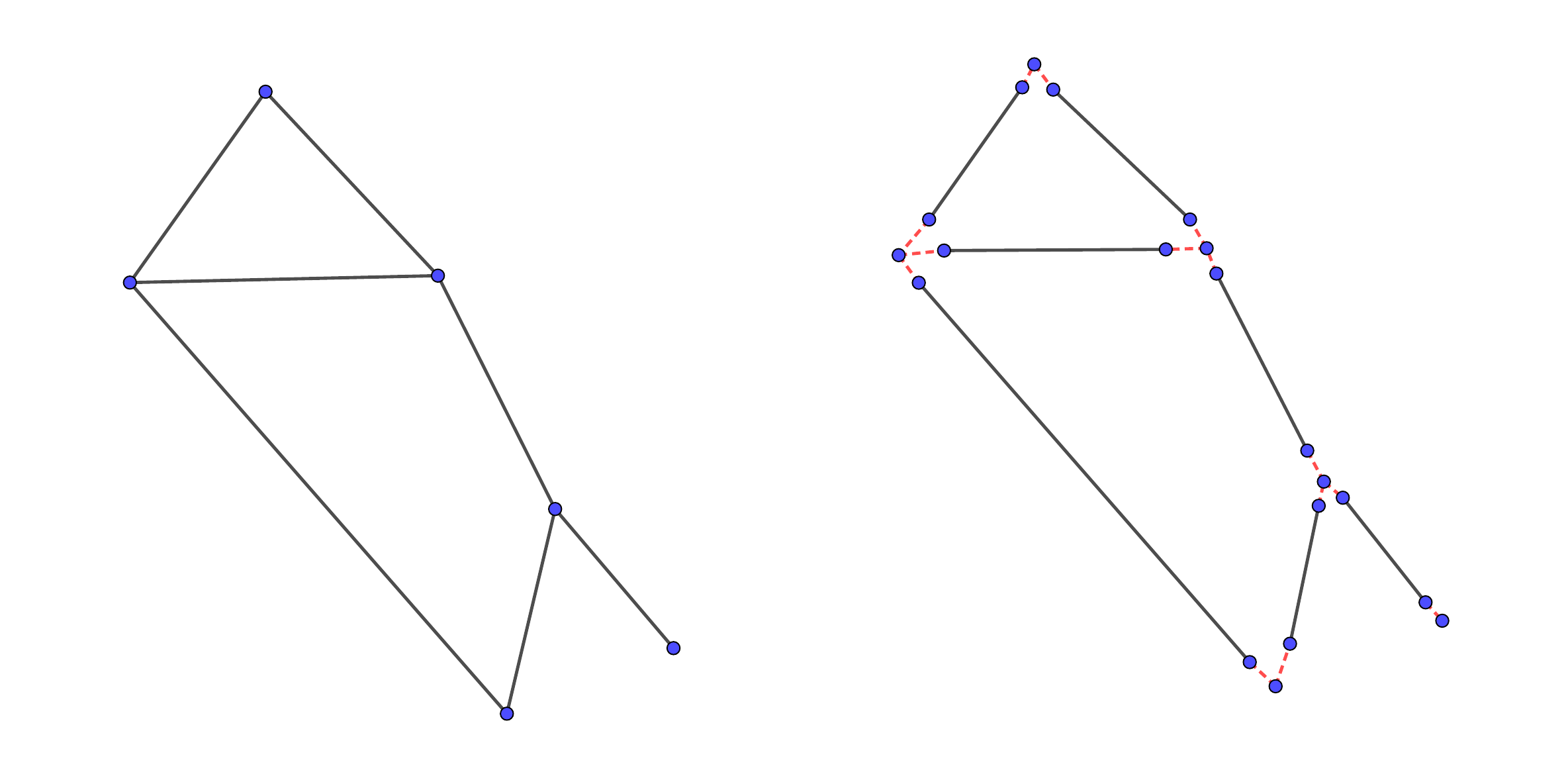}
\centering
\caption{A geometric network on the left, and its expansion on the right. In the expansion, red edges represent edges between limit nodes and value nodes. Black edges correspond to open intervals in the original geometric network.}
\label{fig1}
\end{figure}

In order to perform the embedding of $G$ into $\R$, we apply a slight modification of the technique in depth-first search fused lasso \cite{padilla}. The idea is this: traverse the nodes of the expanded network according to depth-first search, starting at some arbitrary root node. Glue edges together according to the order in which they are visited in the depth-first search. Each of the intervals will be traversed, according to depth-first search. For any function the total variation of the resulting univariate embedding never exceeds twice that of the graph-induced total variation. We formalize this result in the following theorem.

\begin{theorem} \label{TVBounds}
Let $G$ be a connected geometric network, and $\gamma:G \to \R$ be the embedding of $G$ into $\R$ according to depth-first search. Then
\begin{enumerate}[(i)]
\item Each edge of the original (non-expanded) graph is traversed.
\item $\TV(g \circ \gamma^{-1}) \leq 2 \TV_{G}(g)$ for all $g: G \to \R$.
\item Because we use them simultaneously, we dnote $\mu$ and $\mu_G$ denote the Lebesgue  and base measure on $\R$ and $G$, respectively. For any function $f: G \to \R$ and set $A \subseteq \R$,
$$\int_{\gamma^{-1}(A)} f \, d\mu_{G} = \int_{A} f \circ \gamma^{-1} \, d\mu.$$
It follows that for a random variable $X$ on $G$ with density $f_{0}$, $\gamma(X)$ has density $f \circ \gamma^{-1}$. Furthermore, for any functions $f$ and $f_{0}$ on $G$, $h(f, f_{0}) = h(f \circ \gamma^{-1}, f_{0} \circ \gamma^{-1})$.
\end{enumerate}
\end{theorem}
\begin{proof}
For (i), assume for contradiction that there is an open interval of the network $G$ which is not traversed in the depth-first search of the expanded network. Because the degree of a limit node is two,  a limit node must have been a leaf of the DFS spanning tree. But this cannot be. Indeed, one of the limit nodes of the open interval must have been reached first in the depth-first search. Because limit nodes have degree two, when DFS reached that limit node it would proceed across the edge, contradicting that the open interval was not traversed.

For (ii), consider two nodes visited consecutively in DFS of the expanded graph: $\tau(i)$ and $\tau(i+1)$, the $i$th and $i+1$th nodes visited, respectfully. There are two cases to consider. First, assume that $\tau(i)$ is not a leaf of the DFS tree. This implies there is an edge $e$ such that $\TV(g|_{\tau(i)}^{\tau(i+1)}) = \TV(g|_{e})$. For the other case, assume that $\tau(i)$ is a leaf of the DFS tree. From (i) we know that $\tau(i)$ is not a limit node. And because every limit node has degree two, we have that $\tau(i+1)$ is a limit node. Hence the univariate total variation between $\tau(i)$ and $\tau(i+1)$ is $\abs{g(\tau(i+1))-g(\tau(i))}$. Furthermore, there is a path $\pi$, traversed by DFS, such that $\pi$ starts at $\tau(i)$ and ends at $\tau(i+1)$. This requires that the network $G$ be connected, so that the path $\pi$ is a subset of the graph $G$. According to the triangle inequality,
$$\TV(g|_{\tau(i)}^{\tau(i+1)}) \leq \TV_{G}(g|_{\pi}).$$ 

We next use the following fundamental property of DFS (see for example, \cite{cormen}): DFS visits each edge exactly twice. In other words, each edge in $G$ can occur as a member of $\pi$ at most twice. This gives that
$$\TV(g) = \sum_{i} \TV(g|_{\tau(i)}^{\tau(i+1)}) \leq \sum_{i} \TV_{G}(g|_{\pi}) \leq 2 \sum_{e \in E} \TV_{G}(g|_{e}).$$

For (iii), let $f: G \to \R$, and $A \subseteq \gamma(G) \subset \R$. Then
$$\int_{\gamma^{-1}(A)} f \, d\mu_{G} = \sum_{e \in E} \int_{\gamma^{-1}(A) \cap e} f \, d\mu_{G}.$$
On each edge $e$, $\gamma$ is the identity and $\mu_{G} = \mu$. Therefore,
$$\sum_{e \in E} \int_{\gamma^{-1}(A) \cap e} f \, d\mu_{G} = \sum_{e \in E} \int_{A \cap \gamma(e)} f \circ \gamma^{-1} \, d\mu = \int_{A} f \circ \gamma^{-1} \, d\mu.$$
The remaining claims follow from this result. For any random variable $x$ on $G$,
$$\bP(\gamma(x) \in A) = \bP(x \in \gamma^{-1}(A)) = \int_{\gamma^{-1}(A)} f \, d\mu_{G} = \int_{A} f \circ \gamma^{-1} \, d\mu.$$
Therefore $f \circ \gamma^{-1}$ is the density of $\gamma(x)$. Similarly, have that $h(f, f_{0}) = h(f \circ \gamma^{-1}, f_{0} \circ \gamma^{-1})$ because
$$\int_{G} (\sqrt{f} - \sqrt{f_{0}})^2 \, d\mu_{G} = \int_{\gamma^{-1}(\gamma(G))} (\sqrt{f} - \sqrt{f_{0}})^{2} \, d\mu_{G} = \int_{\gamma(G)} (\sqrt{f \circ \gamma^{-1}} - \sqrt{f_{0} \circ \gamma^{-1}})^2 \, d\mu.$$
\end{proof}

\section{Appendix: Bracketing Entropy Results} \label{AppB}
The primary result in this appendix is the following.
\begin{theorem} \label{BrackEnt}
Let $\mathcal{P}_{M}$ be the set of functions $\{p_{f}: f \in \mathcal{F}, \, J(f) \leq M\}$. For some constant $A$, the bracketing entropy of $\mathcal{P}$ satisfies 
$$H_{B}(\delta, \mathcal{P}_{M}, P) \leq A \cdot \frac{M}{\delta}, \; \; \; \; \forall \delta >0.$$
\end{theorem}

The proof of this result is decomposed into the following lemmas. In Lemma \ref{bound} we show that $\mathcal{P}_{M}$ is uniformly bounded, has nonnegative and nonpositive values, and has uniformly bounded total variation. In Lemma \ref{brackent}, we show that any set of functions satisfying these properties is sufficient for the conclusion in Theorem \ref{BrackEnt}. This gives the result for $\mathcal{P}_{M}$.

\begin{lemma}\label{bound}
The set of functions $\{p_{f} : J(f) \leq M\}$ has total variation uniformly bounded by $M/2$, and each function in the set takes nonnegative and nonpositive values. Furthermore, $J(f) \leq M$ gives that $\norm{p_{f}}_{\infty} \leq M/2$.
\end{lemma}
\begin{proof}
The assertion about total variation follows from Lemma \ref{claimslemma}. We also have that each function takes both nonpositive and nonnegative values. Indeed, consider $p_{f}$. From its definition
$$p_{f} = \frac{1}{2} \log\left( \frac{f +f_0}{2 f_{0}} \right).$$
For each $f$, the fact that both $f$ and $f_0$ integrate to $1$ give that for some point $\underline{x} \in \mathcal{X}$ $f(\underline{x}) \leq f_{0}(\underline{x})$. We then have
$$p_{f}(\underline{x}) = \frac{1}{2} \log \left( \frac{f(\underline{x}) + f_{0}(\underline{x})}{2 f_{0}(\underline{x})} \right) \leq 0.$$
Similarly, there exists $\bar{x} \in \mathcal{X}$ such that $f(\bar{x}) \geq f_{0}(\bar{x})$. We then have that $p_{f}(\bar{x}) \geq 0$.

The last claim follows by combining both of the above: a function which takes both nonnegative and nonpositive values and has total variation bounded by $M/2$, is bounded by $M/2$ itself. 
\end{proof}

\begin{lemma} \label{claimslemma} We have the following results.
\begin{enumerate}
\item For any constant $a \geq 0$ and any function $f$, $\TV(\log(f(x) + a)) \leq \TV(\log(f(x)))$. \label{claim1}
\item $J(p_{f}) \leq M$ gives that $\TV(p_{f}) \leq \frac{M}{2}$.
\end{enumerate}
\end{lemma}
\begin{proof}
The first claim is intuitive because the derivative of $\log$ is strictly decreasing. For the proof, consider any two points $x_1, x_2$ in a compact interval $I$. Let $f$ be a real-valued function on $I$. Consider $\abs{\log(f(x_{2})+a) - \log(f(x_{1})+a)}$ for some $a \geq 0$. Without loss of generality, assume that $f(x_{2}) \geq f(x_{1})$. Then
\begin{align} \label{sumBound}
& \abs{\log(f(x_{2})+a) - \log(f(x_{1})+a)} = \log(f(x_{2})+a)-\log(f(x_{1})+a) \nonumber\\
&= \log\left(\frac{f(x_{2})+a}{f(x_{1})+a}\right) \nonumber\\
&\leq \log\left( \frac{f(x_{2})}{f(x_{1})}\right).
\end{align}
Total variation is defined as the supremum over all point partitions $P$, in the interval $I$, of the following sum
$$\TV(g) = \sup_{P} \sum_{x \in P} \abs{g(x_{i+1})-g(x_{i})}.$$
In computing $\TV(\log(f(x)+a))$, we bound each of the terms in the summand with \eqref{sumBound}, to conclude that $\TV(\log(f(x)+a)) \leq \TV(\log(f(x)))$. This gives us the first claim.

For the second claim, we use the following facts about total variation: $\TV(f+g) \leq \TV(f) + \TV(g)$, $\TV(-f) = \TV(f)$, and $\TV(c) = 0$ for any constant $c$. Using these and the first claim, we have
\begin{align*}
\TV(p_{f}) &= \TV\left(\frac{1}{2}\log\left(\frac{f+f_{0}}{2f_{0}}\right)\right) \\ 
& = \TV\left(\frac{1}{2} \log \left(\frac{f}{2f_{0}} + \frac{1}{2}\right) \right) \;\;\;\; \text{ so by the first claim} \\
&\leq \TV\left(\frac{1}{2}\log\left(\frac{f}{2f_{0}}\right)\right) \\
&=\TV(\frac{1}{2}\log(f) - \frac{1}{2} \log(2f_{0})) \\
&\leq \frac{1}{2} \TV(\log(f)) +\frac{1}{2}\TV(\log(2f_{0})) \\
&\leq \frac{1}{2}\TV(\log(f)) + \frac{1}{2}\TV(\log(f_{0}) + \frac{1}{2} \TV(\log(2))\\
&\leq \frac{J(f)}{2}.
\end{align*}
This gives the second claim.
\end{proof}

We next have a lemma for the bracketing entropy of monotone classes of functions, which we will relate to functions of bounded variation. Denote the \emph{bracketing number} of the function class $\F$ with bracketing width $\epsilon$ and metric $d: \mathcal{F} \times \mathcal{F} \to \R$ by $N_{[]}(\epsilon, \mathcal{F}, d)$.

\begin{lemma}[\cite{vandervaart}, Theorem 2.7.5] \label{MonoBrack}
For every probability measure $Q$, there exists a constant $A$ such that the bracketing of monotone functions $f: \R \to [0,1]$ satisfies
$$\log N_{[]}(\epsilon, \cF, L_{2}(Q)) \leq K \left(\frac{1}{\epsilon}\right).$$
\end{lemma}

\begin{lemma}\label{diff}
Let $f: [a,b] \to \R$ be a function such that $\TV(f) \leq k$, and there are $\bar{x}$ and $\underline{x}$ in $[a,b]$ such that $f(\bar{x}) \geq 0$ and $f(\underline{x}) \leq 0$. Then $f$ can be represented as the difference of two nondecreasing functions $g, h$ with $\TV(g)$ and $\TV(h)$ bounded by $k$. Furthermore, for all $x \in [a,b]$,
$$-k \leq g(x) \leq k \text{   and   } -k \leq h(x) \leq k.$$
\end{lemma}
\begin{proof}
Denote by $\TV_{x'}^{x''}(f)$ the total variation of $f$ on $[x', x'']$. Define
$$g(x) := \frac{f(x)+\TV_{a}^{x}(f)}{2}, \; \; h(x) := \frac{\TV_{a}^{x}(f) - f(x)}{2}.$$
Of course, $f = g -h$.

Let $x'' > x'$. Then
\begin{equation}\label{eq1}
g(x'') - g(x') = \frac{f(x'')-f(x') + \TV_{a}^{x''}(f)-\TV_{a}^{x'}(f)}{2}
\end{equation}
and
\begin{equation}\label{eq2}
h(x'') - h(x') = \frac{\TV_{a}^{x''}(f) - \TV_{a}^{x'}(f) - (f(x'') - f(x'))}{2}.
\end{equation}
We have that
$$\TV_{a}^{x''}(f) - \TV_{a}^{x'}(f) = \TV_{x'}^{x''}(f) \geq \abs{f(x'') - f(x')}$$
which allows us to conclude that \eqref{eq1} and \eqref{eq2} are positive. Hence $g$ and $h$ are nondecreasing. Lastly,
$$\TV_{a}^{b}(g) = \frac{f(b) + \TV_{a}^{b}(f) - f(a) + 0}{2} = \frac{f(b) -f(a) + \TV_{a}^{b}(f)}{2} \leq \TV_{a}^{b}(f) \leq k$$
and
$$\TV_{a}^{b}(h) = \frac{\TV_{a}^{b}(f)-f(b) +(0 - f(a))}{2} = \frac{\TV_{a}^{b}(f) +f(a) - f(b)}{2} \leq \TV_{a}^{b}(f) \leq k$$
We have shown the total variation bounds. 

The inequality in the statement of the lemma follows from the nondecreasing nature of these functions. From this property, we have
\begin{equation}\label{eq3}
\frac{f(a)}{2} = g(a) \leq g(x) \leq g(b) = \frac{k-f(b)}{2}
\end{equation}
and 
\begin{equation} \label{eq4}
\frac{-f(a)}{2} = h(a) \leq  h(x) \leq h(b) = \frac{k-f(b)}{2}.
\end{equation}
Because of the fact the assumptions on $\bar{x}$ and $\underline{x}$, $\abs{f(a)} \leq \max\{\abs{f(a) - f(\underline{x})}, \abs{f(a) - f(\bar{x})}\} \leq \TV_{a}^{b}(f) = k$. The same is true of $f(b)$. The conclusion follows by substituting these inequalities into \eqref{eq3} and \eqref{eq4}.
\end{proof}

\begin{lemma}\label{brackent}
Let $\mathcal{G}$ be a set of functions each of which has nonnegative and nonpositive values and have total variation bounded by $M$. Let $Q$ be a probability measure. The bracketing entropy of $\mathcal{G}$ grows like $\frac{1}{\epsilon}$. That is, for some constant $K$ not depending on $Q$,
$$\log N_{[]}(\epsilon, \mathcal{G}, L_{2}(Q)) \leq K (\frac{M}{\epsilon})$$
\end{lemma}

\begin{proof}
Consider the set of functions $\bar{\mathcal{G}} = \frac{1}{M}\mathcal{G}$. From Lemma \ref{bound}, the set $\bar{\mathcal{G}}$ maps from $\R$ to $[-1,1]$, and has total variation bounded by $1$. By Lemma \ref{diff}, $\bar{\mathcal{G}} \subseteq \mathcal{H} - \mathcal{F}$, where each $\mathcal{H}$ and $\mathcal{F}$ contain monotone functions which map $\R \to [-1,1]$. By Lemma \ref{MonoBrack}, the classes $\bar{\mathcal{H}}:= \frac{1}{2} \mathcal{H} + \frac{1}{2}$ and $\bar{\mathcal{F}}:= \frac{1}{2} \mathcal{F}+\frac{1}{2}$ each have bracketing numbers of the form $L^{C_{1}/\epsilon}$ and $L^{C_{2}/\epsilon}$ for constants $L, C_{1}$, and $C_{2}$. We can form an $\epsilon$-bracket of $\mathcal{G}$ from all pairs of $\epsilon/4M$ brackets of $\bar{\mathcal{H}}$ and $\bar{\mathcal{F}}$.

$$g = M \bar{g} = M (h-f) = M 2(\bar{h}-\frac{1}{2} - (\bar{f} - \frac{1}{2})) = 2M(\bar{h} - \bar{f}).$$
Access to an $\epsilon/2M$ bracketing cover of $\bar{\mathcal{H}}$ and $\bar{\mathcal{F}}$ gives functions $l,u,a,b$ such that
$$l \leq h \leq u, \;\;\; a \leq f \leq b$$
and both $\norm{l-u}_{L_{2}(Q)}$ and $\norm{b-a}_{L_{2}(Q)}$ are less than $\epsilon/4M$. We then have
$$2M(l-b) \leq g \leq 2M(u-a).$$
We have formed a bracket in $\mathcal{G}$ of the form $(2M(l-b), 2M(u-a))$. These form an $\epsilon$ bracket because
$$\norm{2M(l-b)-2M(u-a)}_{L_{2}(Q)} \leq 2M \norm{l-u} + 2M\norm{a-b} < \epsilon.$$

There are $L^{4 M C_{1}/\epsilon} \times L^{4 M C_{2}/\epsilon}$ such brackets, so the bracketing entropy satisfies
$$\log N_{[]}(\epsilon, \mathcal{G}, L_{2}(Q)) \leq 4(C_{1}+C_{2}) \log(L) \frac{M}{\epsilon}.$$
This gives the result.
\end{proof}
\end{document}